\documentclass[journal]{IEEEtran}
\makeatletter
\def\ps@headings{%
\def\@oddhead{\mbox{}\scriptsize\rightmark \hfil \thepage}%
\def\@evenhead{\scriptsize\thepage \hfil\leftmark\mbox{}}%
\def\@oddfoot{}%
\def\@evenfoot{}}
\makeatother
\usepackage{latexsym}
\usepackage{amssymb}
\usepackage{epsfig}
\usepackage{amsthm} 
\usepackage[hang]{subfigure}
\usepackage{graphics,color}
\usepackage{amsmath}
\usepackage{algorithm}
\usepackage{algorithmic}
\usepackage{cite}
\usepackage{comment}
\usepackage{url}    
\usepackage{xspace}
\usepackage{psfrag}
\usepackage{alltt}
\usepackage{flushend}

\newtheorem{theorem}{Theorem}
\newtheorem{lemma}{Lemma}

\theoremstyle{claim}
\newtheorem{claim} {Claim}
\theoremstyle{definition}
\newtheorem{fact} {{\bf{Fact}}}

\theoremstyle{definition}
\newtheorem{definition} {Definition}

\newtheorem{remark}{Remark}

\DeclareMathOperator*{\argmin}{arg\,min}

\IEEEoverridecommandlockouts
\begin{document}
%\onehalfspacing
%\onehalfspacing
\title{Opportunistic Routing with Congestion Diversity\\ in Wireless Ad-hoc Networks}
\author{ \IEEEauthorblockN{   A. A. Bhorkar, M. Naghshvar, T. Javidi} \\
           	\IEEEauthorblockA{Department of Electrical Engineering, \\ University of California San Diego, La Jolla, CA 92093}\\
           \url{abhorkar@ucsd.edu}, \url{naghshvar@ucsd.edu}, \url{tjavidi@ucsd.edu}            
          }

 \maketitle
\begin{abstract}
We consider the problem of routing packets across a multi-hop network consisting of multiple sources of traffic and wireless links while ensuring bounded expected delay. Each packet transmission can be overheard by a random subset of receiver nodes among which the next relay is selected opportunistically. 
 The main challenge in the design of minimum-delay routing policies is balancing the trade-off between routing the packets along the shortest paths to the destination and distributing traffic according to the maximum backpressure. Combining important aspects of shortest path and backpressure routing, this paper provides a systematic development of a distributed opportunistic routing policy with congestion diversity (\mbox{D-ORCD}). 
 
 \mbox{D-ORCD} uses a measure of draining time to opportunistically identify and route packets along the paths with an expected low overall congestion.  \mbox{D-ORCD} is proved to ensure a bounded expected delay for all networks and under any admissible traffic. Furthermore, this paper proposes a practical implementation which empirically optimizes critical algorithm parameters and their effects on delay as well as protocol overhead. Realistic Qualnet simulations for 802.11-based networks demonstrate a significant improvement in the average delay over comparative solutions in the literature.  %Finally, various practical modifications to \mbox{D-ORCD} are proposed and their performance are evaluated.

\end{abstract}
\section{Introduction}

Opportunistic routing for multi-hop wireless ad-hoc networks has long been proposed to overcome deficiencies of conventional routing \cite{Larson01,Zorzi03,Morris05,Lott06,Das05}. 
%The opportunistic routing decisions are made in an online manner by choosing the next relay based on the actual transmission outcomes as well as rank ordering of neighboring nodes.    
 Opportunistic routing mitigates the impact of poor wireless links by exploiting the broadcast nature of wireless transmissions and the path diversity. More precisely, the routing decisions are made in an online manner by choosing the next relay based on the actual transmission outcomes as well as a rank ordering of neighboring nodes.
 
 The authors in \cite{Lott06} provided a Markov decision theoretic formulation for opportunistic routing and a unified framework for many versions of opportunistic routing \cite{Larson01,Zorzi03,Morris05}, with the variations due to the authors' choices of costs. 
 In particular, it is shown that for any packet, the optimal routing decision, in the sense of minimum cost or hop-count, is to select the next relay node based on an index.
 This index is equal to the expected cost or hop-count of relaying the packet along the least costly or the shortest feasible path to the destination. 
When multiple streams of packets are to traverse the network, however, it might be necessary to route some packets along longer or more costly paths, if these paths eventually lead to links that are less congested. More precisely, and as noted in \cite{parul07,Neely09}, 
the opportunistic routing schemes in \cite{Larson01,Zorzi03,Morris05,Das05,Lott06} can 
potentially cause severe congestion and unbounded delays (see examples given in \cite{parul07}). %In other words, these routing schemes are said to fail to stabilize otherwise stabilizable traffic. 
In contrast, it is known that an opportunistic variant of  backpressure \cite{TassEph92}, diversity backpressure routing (DIVBAR) \cite{Neely09} ensures bounded expected total backlog for all stabilizable arrival rates. 
To ensure throughput optimality (bounded expected total backlog for all stabilizable arrival rates), backpressure-based algorithms \cite{TassEph92,Neely09} do something very different from \cite{Larson01,Zorzi03,Morris05,Das05,Lott06}: rather than using any 
metric of closeness (or cost) to the destination, they choose the receiver with the largest positive differential backlog 
(routing responsibility is retained by the transmitter if no such receiver exists). 
%That is, the receiver with the smallest queue is chosen if it is less than the source queue, otherwise the source retains the packet and retransmits. 
This very property of ignoring the cost to the destination, however, becomes the bane of this approach, leading to poor delay performance  in low to moderate traffic (see \cite{parul07}).  Other existing provably throughput optimal routing policies \cite{Sarkar08,Xi06,Hassibi08,Yi07}
distribute the traffic locally in a manner similar to DIVBAR and, hence, result in large delay.  

Recognizing the shortcomings of the two approaches, researchers have begun to propose solutions which combine elements of shortest path and backpressure computations \cite{Neely09,Neely11,Ying11}. 
In \cite{Neely09},  E-DIVBAR is proposed: when choosing the next relay among the set of potential
forwarders, E-DIVBAR considers the sum of the differential backlog and the
expected hop-count to the destination (also known as ETX). However, as
shown in \cite{parul07},  E-DIVBAR does not necessarily result in a better delay
performance than DIVBAR.
Instead of a simple addition used in EDIVBAR, this paper provides a distributed opportunistic routing policy with congestion diversity (\mbox{D-ORCD}) under which the congestion information is integrated with the distributed shortest path computations of \cite{Lott06}. In our previous work \cite{Naghshvar09}, ORCD, a centralized version of \mbox{D-ORCD}, is shown to be throughput optimal without discussion on system implications. In this paper, we extend the throughput optimality proof for the distributed version and discuss implementation issues in detail.
 We also tackle some of the system level issues observed in realistic settings via detailed Qualnet simulations. We then show that \mbox{D-ORCD} exhibits better delay performance than state of the art routing policies, 
namely, EXOR, DIVBAR and E-DIVBAR.  

Before we close, we emphasize that some of the ideas 
behind the design of D-ORCD have also been used as 
guiding principles in many routing solutions: some in opportunistic context\cite{Leonardi07, Ying11} and some in conventional context \cite{Neely11}. Below, we detail the similarity and differences between these solutions and our work for the sake of completeness, even though, in our study, we have chosen to focus only on solutions with comparable overhread and similar degree of 
practicality. 
 In \cite{Leonardi07}, perhaps the most related work to ours, the authors
consider a flow-level model of the network and propose a routing policy
referred to as \emph{min-backlogged-path} routing, under which the flows are
routed along the paths with minimum total backlog. 
In this light, \mbox{D-ORCD} can be viewed as a packet-based version of the min-backlogged-path 
routing without a need for the enumeration of paths across the network and costly computations
of total backlog along paths. 
In \cite{Ying11}, authors propose a modified version of backpressure which uses the shortest path information to minimize the average number of hops per packet delivery, while keeping the queues stable. 
%However, as traffic rate increases, there can be a significant increase in the delay  as the intermediate nodes do not consider end-end congestion while deciding next hop. 
In \cite{Neely11}, a modified throughput optimal backpressure
policy, LIFO-Backpressure, is proposed using LIFO discipline at layer 2.  
 Neither of these approaches lend themselves to practical implementations: \cite{Ying11} requires maintaining large number of virtual queues at each node increasing implementation complexity, while \cite{Neely11}  
uses atypical LIFO scheduler resulting in significant reordering of packets. Furthermore, while LIFO-Backpressure policy guarantees stability 
with minimal queue-length variations,  realistic bursty traffic in large multi-hop wireless networks may result in queue-length variations and unnecessarily high delay.

The paper is organized as follows. In Section~\ref{D-ORCDdesign}, we describe the \mbox{D-ORCD} routing algorithm.  In Section~\ref{ImpIssue}, we discuss various protocol implementation issues of \mbox{D-ORCD}.  Section~\ref{experiment} describes our simulation results in detail, where we compare the performance of various routing policies with \mbox{D-ORCD}. We then discuss theoretical guarantees of \mbox{D-ORCD} in Section~\ref{optimality}. We provide concluding remarks and discuss directions for future research work in Section~\ref{Conclusion}. The appendix contains proofs of throughput optimality of \mbox{D-ORCD} under certain assumptions on the model.

\section{Opportunistic Routing with Congestion Diversity}
\label{D-ORCDdesign}

The goal of this paper is to design a routing policy with improved delay
performance over existing opportunistic routing policies. 
In this section, we describe the guiding principle and the design of Opportunistic Routing with Congestion Diversity (\mbox{D-ORCD}). 
We  propose a time-varying distance vector, which enables the network to route packets through a neighbor with the least estimated delivery time.

 D-ORCD opportunistically routes a packet using  three stages of: (a) transmission, (b) acknowledgment, and (c) relaying.
During the transmission stage, a node transmits a packet. 
During the acknowledgment stage, each node that has successfully received the transmitted packet, sends an acknowledgment (ACK) to the transmitter node. \mbox{D-ORCD} then takes routing decisions based on a 
congestion-aware distance vector metric, referred to as the
\emph{congestion measure}. More specifically,  during the relaying stage, the relaying responsibility of the packet is shifted to a node with the  
least congestion measure among the ones that have received the packet. 
The congestion measure of a node associated with a given destination provides 
an estimate of the best possible draining time of a packet arriving at that node  
until it reaches destination.  
Each node is responsible to update its congestion measure and transmit this information to its neighbors. 
Next, we detail D-ORCD design and the computations performed at each node to update the congestion measure.

\subsection{D-ORCD Design}
\label{design}
We consider a network of $D$ nodes labelled by $\Omega = \{
 1, \ldots,D \}$. Let $p_{ij}$ be the probability that the packet transmitted by node $i$ is successfully received by node $j$.
Node $j$ is said to be \emph{reachable} by node $i$, if $p_{ij} > 0$.
The set of all nodes in the network which are reachable by node $i$ is referred to as \emph{neighborhood} of node $i$ and is denoted by $\mathcal{N}(i)$.  
 %Table~\ref{notations} provides notations used in the description of the algorithm. 

\begin{table}%
\centering
\caption{Notations used in the description of the algorithm}
\begin{tabular}{|c|l|}
\hline 
Symbol  & Definition \\
\hline
 $\mathcal{N}(i)$ & Neighbours of node $i$ \\  
    & \\
 $V^d_i(t)$ &  Congestion measure at node $i$ at time $t$\\
           & \\
 $\tilde{V}_k^{(i,d)}(t)$ & Congestion measure obtained at node $i$ from $k$ \\
          & \\
  $T(t)$  & Ending time of the latest computation cycle before time $t$ \\
         & \\
  $T_c$ & Duration of the computation interval \\ 
    &  \\
     $T_s$ & Control packet transmission interval \\ 
    &  \\
  $L_i(t)$ & Local congestion at node $i$ \\
          & \\  
  $D_i(t)$ & Congestion down the stream for node $i$ \\
          & \\                
 $K_{D-ORCD}^{(i,d)}(t)$ & Selected relay for transmission at node $i$ \\
          & \\        
 $S_i(t)$ & Set of nodes receiving packet transmitted by node $i$ \\
          & \\     
${Q}_i^d(t)$ & Queue-length at node $i$ destined for $d$ at time $t$\\
            &  \\  
 $\bar{Q}_i^d(t)$ & Average queue-length at node $i$ destined for $d$ \\
            &  \\      
$P_{succ-k}^{(i,d)}(t)$      & Probability that highest priority node $k$ receives packet\\
          & \\
 $P^{(i,d)}(t)$      & Probability that at-least one higher \\
 		& priority node receives packet \\   
          & \\
 $H^{(i,d)}(t)$      & Set of higher priority nodes than node $i$ \\         
            \hline
\end{tabular}
\label{notations}
\end{table}

D-ORCD relies on a routing table at each node to determine the next best hop. 
The routing table at node $i$ consists of a list of neighbors $\mathcal{N}(i)$ and 
a structure consisting of estimated
congestion measure for all neighbors in $\mathcal{N}(i)$ associated with different destinations.  
The routing table acts as a storage and decision component at the routing layer. The routing table is updated
using a  ``virtual routing table'' at the end of every ``computational cycle'': an interval $T_c$ units of time. 
To update virtual routing table, during the progression of the computation cycle  the nodes exchange and 
compute the temporary congestion measures. 
The temporary congestion measures are computed in a fashion similar to a distributed stochastic routing computation of~\cite{Lott06} 
using the backlog information at the beginning of the computation cycle
(generalizing the computations of distributed Bellman-Ford). We conceptualize this in terms of a virtual routing table updating
and maintaining these temporary congestion measures. 
We assume that each node has a common global time to ensure that the nodes update the routing table roughly at the same time. 

 We denote the temporary congestion measure associated 
with node $i \in \Omega$ at time $t$ and destinatifon $d \in \Omega$ as $V_i^d(t)$. 
Each node $i$ computes $V_i^d(t)$  based on congestion measures $\tilde{V}_{k}^{(i,d)}(t)$ obtained via periodic communication with its neighbours $k \in  \mathcal{N}(i)$ and the queue backlog at the start of the computation cycle. 
 \mbox{D-ORCD}  stores these temporary congestion measures $\{V_i^d(t)\}_{d \in \Omega}$ and $\{\tilde{V}_{k}^{(i,d)}(t)\}_{d \in \Omega, k \in \mathcal{N}(i)}$ in the virtual  routing table.
More precisely, node $i$ periodically compute its own congestion measure and subsequently  advertises it to its neighbors using control packets at intervals of $T_s \le T_c$ seconds. Finally the actual routing table is updated using the entries 
in the virtual routing table after every $T_c$ seconds. The sequence of operations performed by D-ORCD are shown in Figs.~\ref{fig:rtable1},\ref{fig:rtable}.

Meanwhile, for routing decisions, node $i$ uses the entries in the actual routing routing table (computed during the last computation 
cycle). Let  
 $T(t) = \max_n \{nT_c: nT_c \le t, n \in \mathrm{Z}\}$ be the ending time of the latest computation cycle. 
 Then node $i$ stores $\tilde{V}_k^{(i,d)}(T(t))$ in the actual routing table and selects the next best hop $K_{D-ORCD}^{(i,d)}$ to minimize the packet's draining time, i.e. 
\begin{eqnarray}
K_{{D-ORCD}}^{(i,d)}(t) = \argmin_{k \in S_i(t) \cup i } \tilde{V}_k^{(i,d)}(T(t)), 
\end{eqnarray}
where $S_i(t)$ denotes a random set of nodes  receiving the packet transmitted by node $i$ at time $t$.  

Next, we describe the distributed computations performed during each computation cycle.

\begin{figure}
\centering
\includegraphics[width=.5\textwidth]{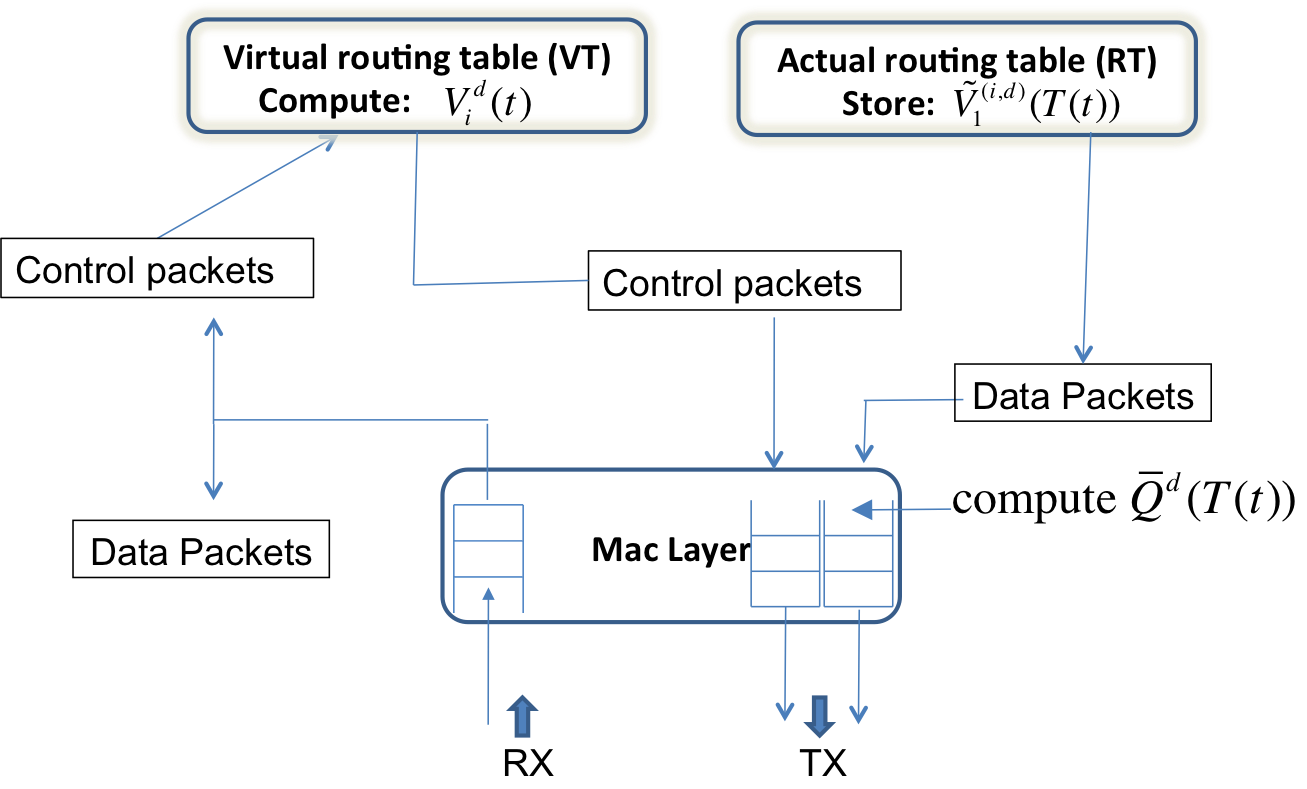}%
\caption{Operation of D-ORCD}
\label{fig:rtable1}
\end{figure}

\begin{figure}
\centering
\includegraphics[width=.5\textwidth]{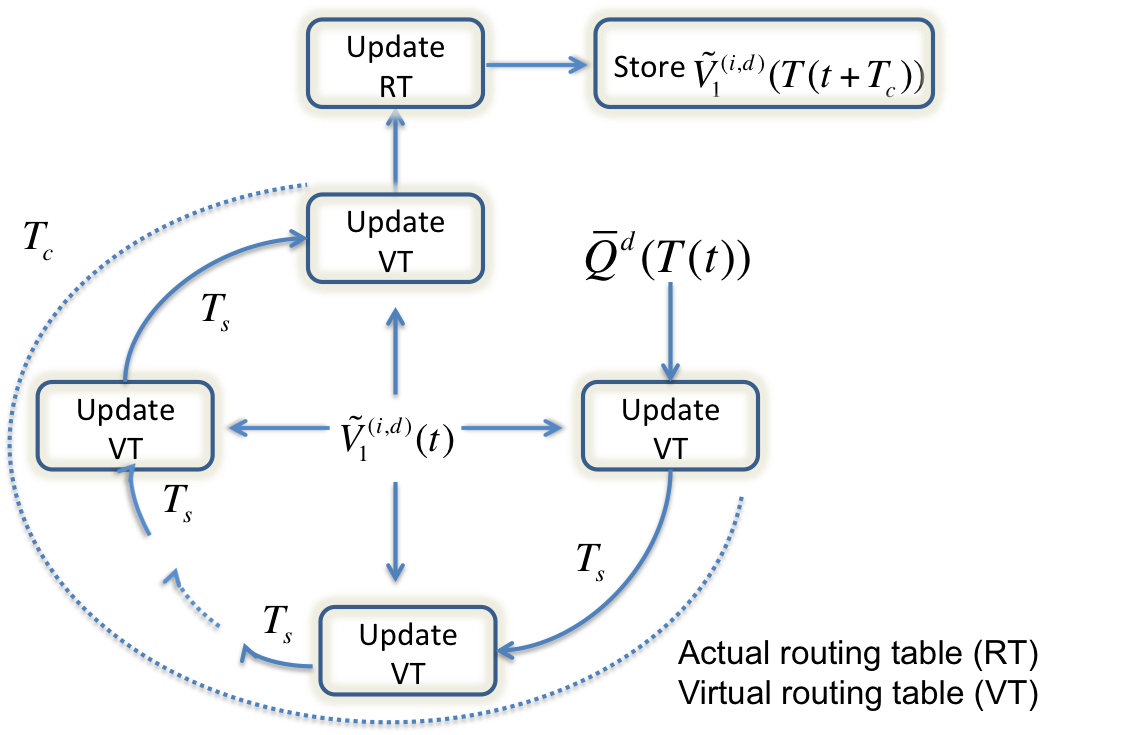}%
\caption{Actual routing table is updated every $T_c$ units of time while virtual routing table is updated after receiving any control packet }
\label{fig:rtable}
\end{figure}

\subsection{Congestion Measure Computations}

The congestion measure associated with node $i$ for a destination $d$ at time $t$ 
is the aggregate sum of the local draining time at node $i$ (denoted by $L_i^d(t)$) and the draining time from its next hop to the 
destination (denoted by $D_i^d(t)$), i.e. 
\begin{eqnarray}
\label{fixpointDec_d}
V_i^d(t) &=&  L_i^d(t) +   D_i^d(t).      
\end{eqnarray} 

Assuming a FIFO discipline at layer-2, we proceed to decompose the local draining time. This relies on the observation that 
when a packet arrives at a node, $i$, 
its waiting time is equal to the time spent in draining the packets 
that have arrived earlier plus its own transmission time. If $P^{(i,d)}(t)$ denotes the probability that the packet transmitted by node $i$ is successfully received by a node with lower congestion measure, then expected transmission time at node $i$ for the packet is given by $\frac{1}{P^{i,d}(t)}$. Let $\bar{Q}_i^{d}(t)$ denote the number of packets destined for destination $d$ averaged over previous computation cycle. 
 $\bar{Q}_i^{d}(t)$ is updated as 
\begin{equation}
\nonumber
\bar{Q}_i^{d}(t) = \frac{T_s}{T_c}\sum_{l=0}^{\frac{T_c}{T_s}-1} {Q}_i^{d}(T(t)-l). 
\end{equation} 
The local draining time for node $i$ to destination $d$ at time $t$ is approximated as,
\begin{equation}
L_i^d(t) = \frac{1}{P^{(i,d)}(t)} + \sum_{d' \in \Omega}  \frac{\bar{Q}_i^{d'}(T(t))}{P^{(i,d')}(t)}. 
\end{equation}

  \mbox{D-ORCD} computes the expected congestion measure ``down the stream'' for each node $i \in \Omega$ using the latest congestion measure $\tilde{V}_k^{(i,d)}(t)$ received  from nodes $k \in \Omega$ with lower congestion measure.  With respect to the destination $d$, a node $k \in \Omega$ is defined as a higher priority node than node $i$ if 
$\tilde{V}_k^{(i,d)}(t) < V_i^d(t)$ and the set of  higher priority nodes as $H^{(i,d)}(t)$. 
   Let  $P_{succ-k}^{(i,d)}(t)$ be the probability that node $k$ is the highest priority node to successfully hear 
node $i$ at time $t$ and $k \in H^{(i,d)}(t)$. 
 As a result, the expected congestion ``down the stream'' $D^d_i(t)$ can be given as  
 \begin{eqnarray}
D^d_i(t) = \sum_{k \in \Omega} P_{succ-k}^{(i,d)}(t) \tilde{V}^{(i,d)}_k(t). 
\end{eqnarray}

\begin{remark} 
In each computation cycle, assuming $T_c$ is large, D-ORCD computations converge to the Bellman equation 
associated with the minimum cost (``shortest path") route in a network, where the link costs
are given in terms of the queue length $\bar{Q}_i^d(t)$. 
\end{remark}
 
\begin{remark}

 If the links success probabilities have independent realizations, then for all $S \subseteq \Omega$, $P(S_i(t)=S) = \prod_{k \in S} \prod_{l \notin S} p_{ik} (1-p_{il})$. The success probabilities $P^{(i,d)}(t)$ and $P^{(i,d)}_{succ-k}(t)$ can be calculated as  
 
 \begin{align}
P^{(i,d)}(t) &=\mathop{\sum_{S: \exists k \in S\cap H^{(i,d)}(t) }}  P(S_i(t)=S), \\
P^{(i,d)}_{succ-k}(t) &= \frac{1}{P^{(i,d)}(t)} \times \mathop{\sum_{S: \tilde{V}^{(k,d)}(t) < \tilde{V}^{(k',d)}(t) }}
_{ k',k \in S \cap H^{(i,d)}(t)} 
P(S_i(t)=S). 
\end{align}

\end{remark}

\subsection{Opportunistic Routing with Partial Diversity}
\label{PD-D-ORCD}
The three-way handshake procedure discussed in Section~\ref{design}  to achieve receiver diversity gain in an opportunistic scheme  is achieved at the cost of an increase in the control overhead. In particular, it is easy to see that this overhead cost, which is the total number of ACKs sent per data packet transmission, increases
linearly with the size of the set of potential forwarders. Thus, we consider a modification of D-ORCD in the form of opportunistically 
routing with partial diversity (P-ORCD). 
This class of routing policies is parametrized by parameter $M$ denoting the maximum number of forwarder nodes. This is equivalent to a constraint on the maximum number of nodes allowed to send acknowledgment per data packet transmission.
%In fact, transmissions from each node can be heard by at most $M$ of its neighbors. 
Such a constraint will sacrifice the diversity gain, and hence the performance of any opportunistic routing policy, for lower communication overhead.

In order to implement opportunistic routing policies with partial diversity, before the transmission stage occurs  we find the set of ``best neighbors'' for each node $i$ at any time $t$, denoted by $B_i^*(t)$, where $|B^*_i(t)| \le M$. 
After transmission of a packet from node $i$ at time $t$, the routing decision is made as follows: 1)~among the nodes in $B_i^*(t) \cap S_i(t)$, select a node with the lowest congestion measure as the next forwarder; or 2) retain the packet if none of the nodes in the set $B_i^*(t)$ has received the packet. 
Next we give a mathematical formulation for modifications of \mbox{D-ORCD} with partial diversity. 

Let $\mathcal{B}$ be the collection of all subsets of $\Omega$ of size less than or equal to $M$, i.e.
$\mathcal{B} = \{ B \subseteq \Omega: |B| \le M \}$.

In \mbox{D-ORCD} protocol with partial diversity, \mbox{(PD-ORCD)},  
the corresponding quantities $\bar{V}^{d}_i(t)$ are updated as
\begin{align}
\bar{V}^{d}_i(t) = \min_{B \in \mathcal{B}} \left \{ L^d_i(t) +  \sum_{k: k \in \mathcal{B}} {P^{(i,d)}_{succ-k}(t)}  \tilde{\bar{V}}^{(i,d)}_k(t) \right \}, 
\end{align}
while the next hop is selected as  
\begin{eqnarray}
{K}_{{PD-ORCD}}^{(i,d)}(t) = \argmin_{k \in \{S_i(t)\cap B\} \cup i } \tilde{\bar{V}}_k^{(i,d)}(T(t)). 
\end{eqnarray}

 We carry out a simulation study for the delay performance of \mbox{D-ORCD} with these modifications and compare it to the delay performance of the other routing policies in Section \ref{experiment}.

\begin{remark}
When $M=1$ each node can send packets only to one of its neighbors. Therefore, this routing policy cannot take the advantage of the broadcast nature of wireless transmissions anymore, and is classified as conventional routing. 
\end{remark}

In the next section, we discuss the practical issues associated with computation of
the time-varying congestion measures $V^d_i(t)$, $i \in \Omega$.
Furthermore, we propose practical implementations and heuristics.

\section{Implementation Details: Protocol Components}
\label{implementation}
\label{ImpIssue}

In this section we discuss the implementation issues of \mbox{D-ORCD} which involves distributed and asynchronous iterative computations of $V_i^d(t)$'s.  We  provide brief discussion of the basic challenges  of \mbox{D-ORCD} including the three-way handshake procedure employed at MAC, link quality estimation, and avoidance of loops while routing.

\subsection{802.11 Compatible Implementation}
\label{80211}
\subsubsection{Three way Handshake}
The implementation of \mbox{D-ORCD}, analogous to any opportunistic routing scheme, involves the selection of a relay node among the  candidate set of nodes that have received and acknowledged a packet successfully. 
One of the major challenges in the implementation of an opportunistic routing algorithm, in general, and \mbox{D-ORCD} in particular, is the design of an 802.11 compatible acknowledgement mechanism at the MAC layer. Below we propose a practical and simple way to implement acknowledgement architecture. 

{The transmission at any node $i$ is done according to 802.11 CSMA/CA mechanism. Specially, before any transmission, transmitter $i$ performs channel sensing and starts transmission after the backoff counter is decremented to zero. }
For each neighbor node $j \in \mathcal{N}(i)$, the transmitter node $i$ then reserves a virtual time slot of duration  $T_{ACK}+T_{SIFS}$, 
where $T_{ACK}$ is the duration of the acknowledgement packet and $T_{SIFS}$ is the duration of Short Inter Frame Space (SIFS) \cite{802.11}. 
Transmitter $i$ then piggy-backs a priority ordering of nodes $\mathcal{N}(i)$ with each data packet transmitted. 
The priority ordering determines the virtual time slot in which the candidate nodes transmit their acknowledgement. 
Nodes in the set $S_i$ that have successfully received the packet then transmit acknowledgement packets sequentially in the order determined by the transmitter node. 
%Fig. \ref{timing} shows a typical sequence of control packets for the topology in Fig. \ref{topo} when node$~0$ piggy-backs the information about the ordering \{1,2\} and when both nodes receive the packet successfully. 

After a waiting time of $T_{wait} = |\mathcal{N}(i)|(T_{ACK}+T_{SIFS})$ during which each node in the set $S_i$ has had a chance to send an ACK, node $i$ transmits a FOrwarding control packet (FO).
The FO packets contain the identity of the next forwarder, which may be node $i$ itself (i.e. node $i$ retains the packet) or any node $j\in S_i$. 
If $T_{wait}$ expires and no FO packet is received (FO packet reception is unsuccessful), then the corresponding candidate nodes drop the received data packet. {If  transmitter $i$ does not receive any acknowledgement, it retransmits the packet.  The  
backoff window is doubled after every retransmission. Furthermore, the packet is dropped
if the retry limit (set to 7) is reached.  }
\psfrag{0}{\scriptsize{1}}
\psfrag{1}{\scriptsize{2}}
\psfrag{2}{\scriptsize{3}}
\psfrag{d}{\scriptsize{4}}
\psfrag{P0}{\scriptsize{$p_{23}$}}
\psfrag{P1}{\scriptsize{$p_{12}$}}
\psfrag{P2}{\scriptsize{$p_{23}$}}
\psfrag{P3}{\scriptsize{$p_{34}$}}

%\begin{comment}

\begin{figure}[ht]
\centering
\includegraphics[width=0.45\textwidth]{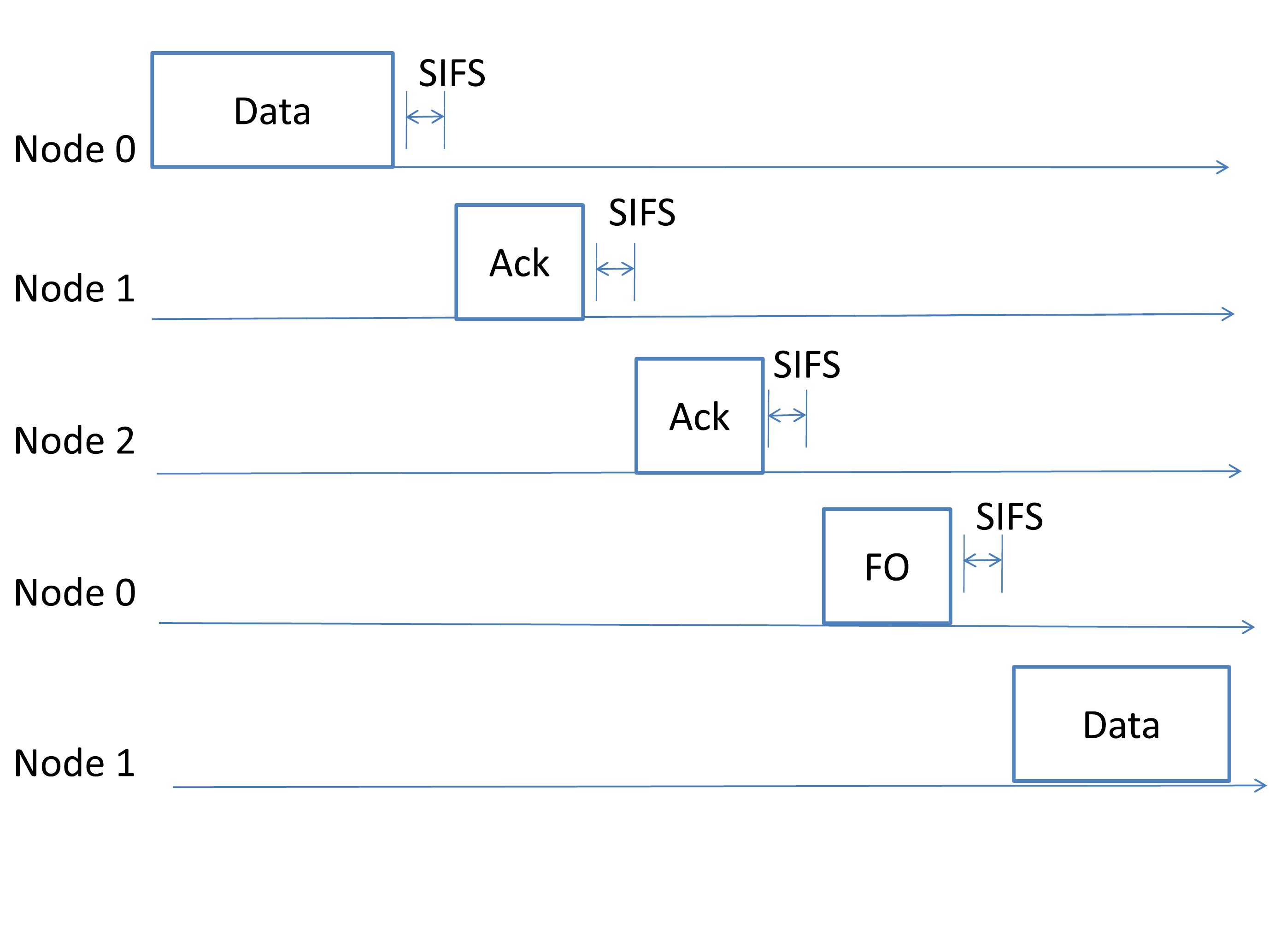}
\caption{Typical packet transmission sequence for \mbox{D-ORCD} when packet transmitted by node 0 is received by nodes 1 and 2.}
 \label{timing}
\end{figure}
%\end{comment}

\subsection{Control Packets fidelity}
\mbox{D-ORCD} depends on a reliable, frequent, and timely delivery of the control packets.  As documented in \cite{Shaikh00routingstability}, the loss of control packets 
may destabilize the algorithm operation and cause significant performance degradation for 
many well known routing algorithms.
In our implementation, we have 
taken advantage of the priority-based queuing to implement this component of the control plane. 
\mbox{D-ORCD} prioritizes the control packets by assigning them the highest strict priority,   
reducing the probability that the packets are dropped at the MAC layer and also ensuring a timely delivery of the control packets. In particular, \mbox{D-ORCD} utilizes priority queues: data packets are assigned to the lower priority queue   and control packets are assigned to the  
higher priority queue. Moreover, \mbox{D-ORCD} scheduler assigns a 
 sufficiently lower PHY rate for the control packets.

\subsection{Link Quality Estimation Protocol}
\label{LPP}
 D-ORCD computations given by (\ref{fixpointDec_d}) utilize link success probabilities $p_{ij}$ for each pair of nodes $i$, $j$. 
We now describe  method to determine the probability of successfully receiving a data packet for each pair of nodes $i,j \in \Omega$. 
It consists of two phases: active probing and passive probing. 
In the active probing, dedicated probe packets are broadcasted periodically to estimate link success probabilities. 
In passive probing, we utilize the overhearing capability of the wireless medium. The nodes are configured to promiscuous mode, hence enabling them to hear the packets from neighbors. In passive probing, the MAC layer keeps track of the number of packets received from the neighbors including the retransmissions. Finally, a weighted average is used to combine the active and passive estimates to determine the link success probabilities. Passive probing does not introduce any additional overhead cost but can be slow, while active probing rate is set independently of the data rate but introduces costly overhead.
  
\subsection{Loop Avoidance Heuristic}
D-ORCD approximates the solution to the fixed point equation via  distributed distance vector approach. The classical problem of counting to infinity \cite{Moss82} in distance vector routing can affect D-ORCD performance due to the
time varying nature of the congestion metric. 
 The problem is most acute when 
there is a sudden burst of traffic.\footnote{Similar to the broken link scenario in a typical distance vector routing.} and can cause 
 severe transient effects due to slow updates of the control packets. The looping results 
in large delays, increased interference and loss of packets.\footnote{Packet loss occurs when  time to live (TTL) value exceeds 
the number of allowed hops (typically 64)}.

 To address this issue, in our experiments 
we utilize an extension of the Split-horizon with poison reverse solution \cite{Split00} to avoid loops.
In Split-horizon with poison reverse, a node advertises routes as unreachable to the node through which they were learned.  We have extended the rule to D-ORCD by advertising the routes as unreachable to higher ranked nodes. 
 This removes most looping routes before they can propagate through the network.
%\end{enumerate}

\section{Simulations}
\label{experiment}
In this section, we compare the expected delay encountered by the packets in the network under various opportunistic routing policies: ExOR, DIVBAR, E-DIVBAR and D-ORCD in Qualnet simulations. 
We first investigate the performance of D-ORCD with respect to a canonical example to demonstrate D-ORCD gains \cite{parul07}. We then use a realistic topology of 16 nodes placed in a grid topology to demonstrate the robust performance improvement in practical settings. 

\subsection{The Simulation Setup}
 Our simulations are performed in QualNet. We consider two set of topologies in our experimental study: 
\begin{enumerate}
\item Canonical Example: In this example, we study the canonical example in Fig.~\ref{fig:ParulExmpl}. We motivate the performance improvement for D-ORCD by a scenario which exemplify the need to avoid congestion in the network by highlighting the shortcomings of the existing routing paradigms: shortest path and backpressure. 
\item Grid Topology: We study an outdoor wireless settings of grid topology consisting of 16 nodes separated by a distance of $200$ meters. These simulations demonstrate a robust performance gain under D-ORCD in a realistic network. 
%\item Los Angeles terrain: In this example, we study a Log-Angeles terrain of 1 mile x 1 mile and 16 nodes are placed randomly. 
\end{enumerate}

We now describe the parameters settings in the simulation. 
The nodes are equipped with 802.11b radios transmitting at 11 Mbps with transmission power 15 dBm. 
The wireless medium model includes Rician fading with K-factor of 4 and Log-normal shadowing with mean 4dB. 
In the canonical example path loss is determined by pathloss matrix which gives the attenuation of the received signal power with distance from the transmitter for every pair of network nodes, while for grid topology the path loss follows ITM model in \cite{Doble96}. The antenna model is the standard omnidirectional antenna model with the default settings of the simulator. The network queues are FIFO with finite buffer of 750 KB.

The acknowledgement packets are short packets of length {24 bytes} transmitted at 11 Mbps,
while FO packets are of length 20 bytes and transmitted at lower rate of 1 Mbps to ensure reliability.
If unspecified, packets are generated according to a poission modulated Markov traffic. The packets are assumed to be of length 512 bytes equipped with simple cyclic redundancy check (CRC) error detection. The control packets are transmitted periodically at an interval of $T_s=0.5$ seconds.  

We have chosen partial diversity $M=4$ and update frequency $T_c=T_s=0.5$ seconds in our experimentations.  
A discussion on the choice of parameters in the design of \mbox{D-ORCD} is provided in Section \ref{choice}.

In our study, we have compared the performance of \mbox{D-ORCD} against the state of the art routing algorithms. 
Before we proceed, we describe these candidate algorithms as well as our implementation of them. 
\begin{itemize}
\item DIVBAR \cite{Neely06}:  We implemented DIVBAR to select the next hop based on a weighted differential backlog. Specifically, let $\tilde{Q}^{(i,d)}_k(t)$ denote the latest information at node $i$ about the number of packets buffered in queue $k$ destined for destination $d$. 
For any destination $d$, DIVBAR chooses the next hop $K_{DIVBAR}^{(i,d)}(t)$, such that
    \begin{eqnarray}
     K_{{DIVBAR}}^{(i,d)}(t)  & = & \argmin_{k \in S_i(t) \cup i} (\tilde{Q}^{(i,d)}_k(t) - Q^d_i(t)). 
    \end{eqnarray}
  We have created virtual queues for each destination to identify differential backlog associated with different destinations. Note that original backpressure algorithm proposed in \cite{Neely06} is done in conjunction with a
  scheduler to maximize the network's overall weighted differential backlog as well as a mechanism to choose destination queue to be served.  
In our implementation, we serve the packets in a prioritized manner based on the destination using 802.11 MAC.  { %\color{blue}
  Specifically, packet with destination $m(t)$ is selected among all possible virtual queues
  such that
  \begin{eqnarray}
     \label{commodity}
      m(t)  & = & \argmin_{d} \{ \min_k (\tilde{Q}^{(i,d)}_k(t) - Q^d_i(t)) \}. 
    \end{eqnarray}
    In order to implement a  priority scheduling we utilize a priority scheduler  
    such that the packet destined for $m(t)$ is assigned higher priority queue. 
  }We have implemented the DIVBAR algorithm using a structure similar to \mbox{D-ORCD} (in which  
${V}_i^d(t)$ is replaced with $Q^d_i(t))$.  
\item ExOR  \cite{Morris05} : ExOR uses ETX metric when routing the packet without considering queuing information at the nodes. Specifically, 
for a packet destined for node $d$, the next hop $K_{ExOR}^{(i,d)}$ is chosen such that
        \begin{equation}
     K_{{ExOR}}^{(i,d)}(t)   =  \argmin_{k \in S_i(t) \cup i} ETX^{(k,d)}, 
    \end{equation}
where $ETX^{(k,d)}$ is the minimum number of transmissions from node $k$ to destination $d$ given by,
            \begin{equation}
            ETX^{(k,d)} =  \min_{j} \big \{ \frac{1}{p_{kj}} +  ETX^{(j,d)} \big \}. 
            \end{equation}    
               We have used our distributed architecture for the calculation of ETX metric by taking $Q_i(t) = 1$ for all $i\in \Omega$ and $M=1$ in the calculation of $V_i(t)$, even though, in principle, the overhead can 
               be held much lower due to the time invariant nature of node ordering. 

\item \mbox{E-DIVBAR}\cite{Neely09}: E-DIVBAR is a variant of DIVBAR, where along with the queue information, ETX metric is used for path selection.  In particular, for a packet destined for $d$, the next hop $K_{{E-DIVBAR}}^{(i,d)}$ is chosen such that 
        \begin{eqnarray}
        \nonumber
     K_{{E-DIVBAR}}^{(i,d)}(t)  &=&  \argmin_{k \in S_i(t) \cup i} \left \{ (\tilde{Q}^{(i,d)}_k(t) - Q^d_i(t)) \right. \\
     		&& + \left.  ETX^{(k,d)} \right \}. 
    \end{eqnarray}
    \mbox{E-DIVBAR} algorithm is also implemented using a structure identical to \mbox{D-ORCD} and \mbox{DIVBAR}, however, the control packets contain information about the queue-length as well as the ETX for a given destination. The commodity selection is performed using the same equation (\ref{commodity}) as DIVBAR. 
 \end{itemize}

Next, we study the canonical example where we compare the average delay encountered by packets in the network under various routing policies: \mbox{ExOR}, \mbox{DIVBAR}, \mbox{E-DIVBAR} and \mbox{D-ORCD}. The choice of the canonical network enables us to clearly reveal the high capability of D-ORCD in balancing the traffic taking advantage of path diversity in the network.

\subsection{Performance of D-ORCD: Canonical Example}

\begin{figure}
\centering
\includegraphics[width=.45\textwidth]{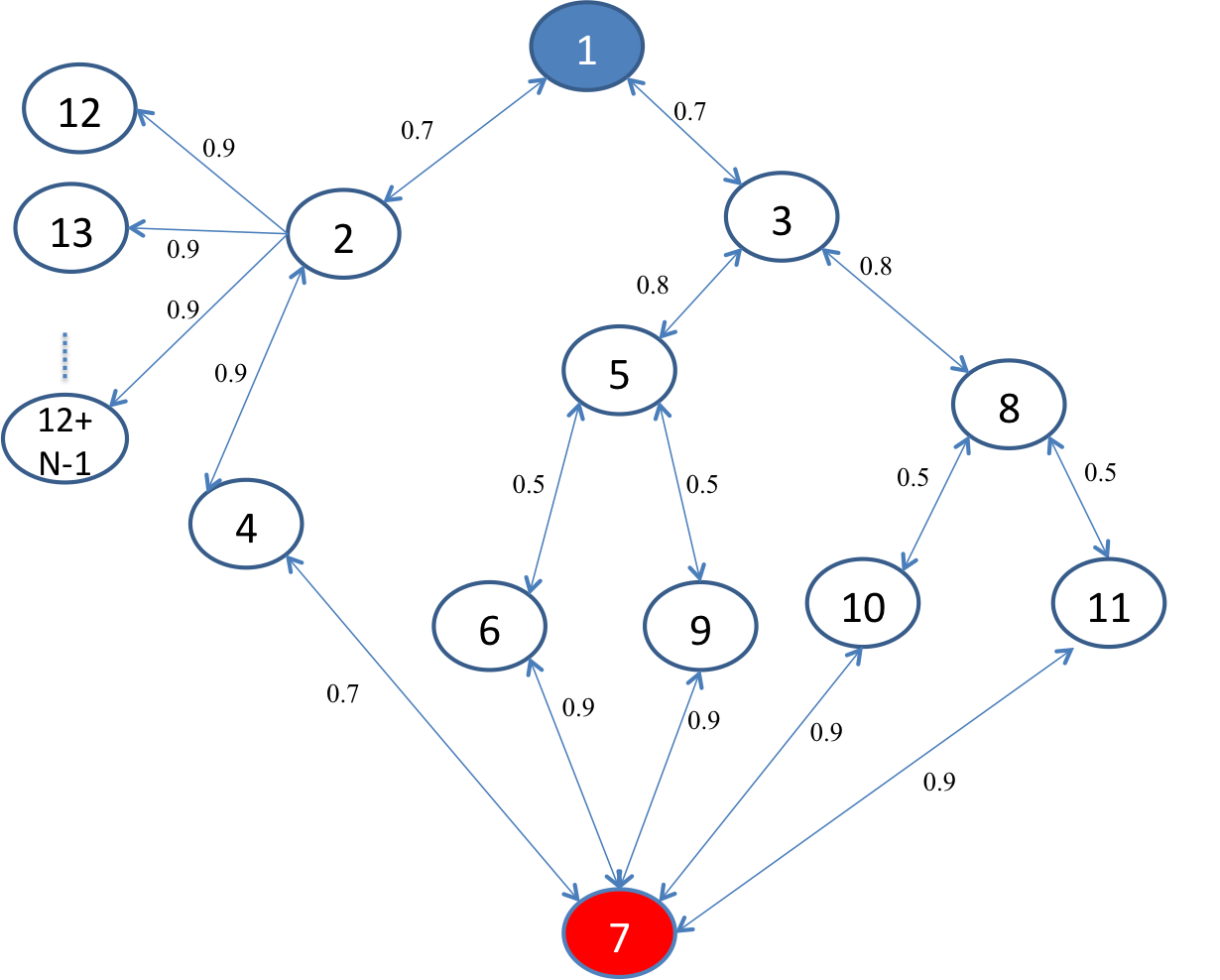}%
\caption{Structure of the canonical network from \cite{parul07}. The fractions
        on the links show the probability of successful transmission on each link. }
\label{fig:ParulExmpl}
\end{figure}

\begin{figure*}[ht]
\centering
\subfigure[Delay] {
\includegraphics[width=.45\textwidth]{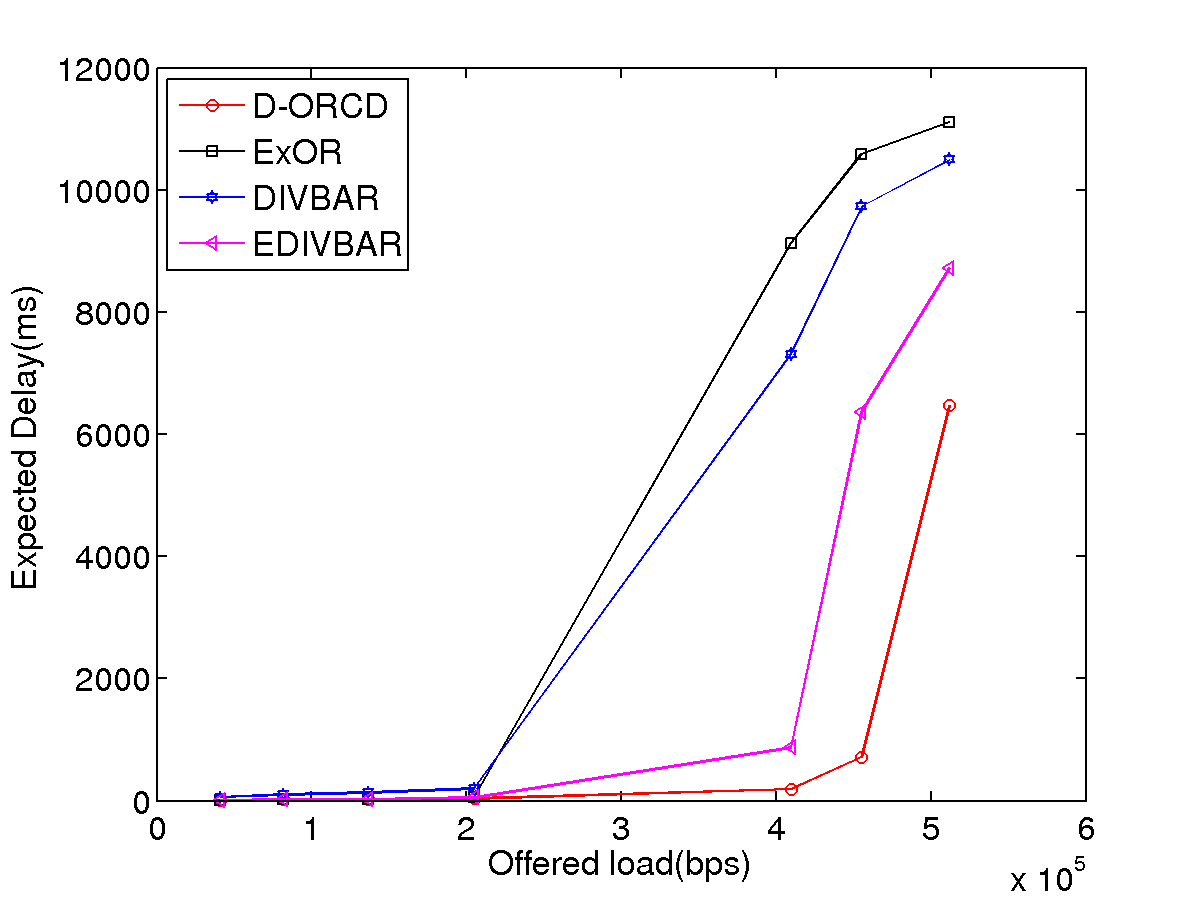}%
 }
\subfigure[Fraction of packet loss]{
\includegraphics[width=0.45\textwidth]{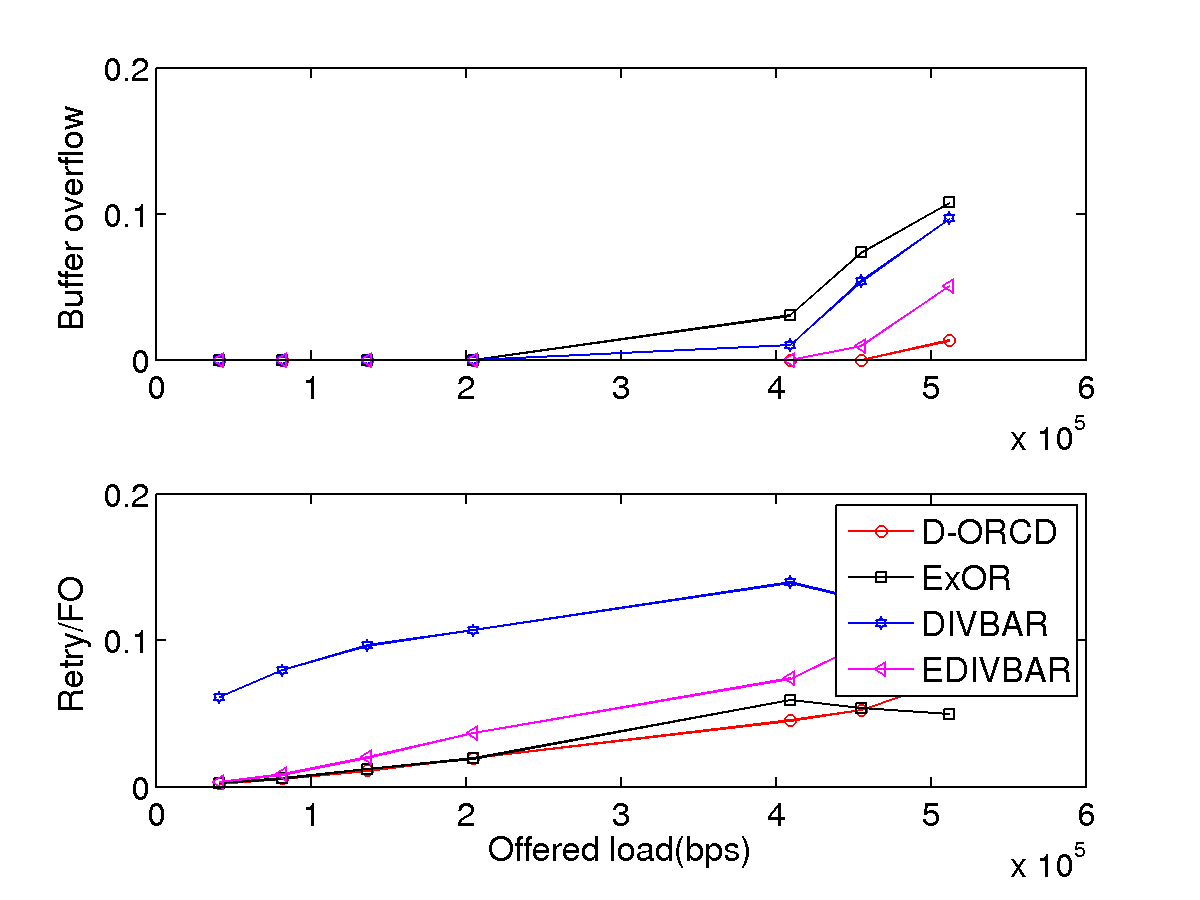}               
   \label{fig:thru}
}
\caption{Performance for Canonical Example for  $N$=2}
\label{fig:perfN2}
\end{figure*}

Consider the  network shown in Fig.~\ref{fig:ParulExmpl} which is parameterized
by  $N$. 
Nodes $12,13, \ldots, 12+(N-1)$ form a ``hole'' in the network whose size is controlled by the parameter $N$.
We now discuss the delay gains under D-ORCD 
 as parameters $N$ and $\lambda_1$ (the incoming traffic rate at node $1$) are varied and verify them in this section.

Note that the source node $1$ can route packets either through node $2$ or node $3$.
Since only node $1$ has a routing choice we focus on the delay experienced by packets originating in node $1$. 
Fig.~\ref{fig:perfN2} provide plots of the average end-to-end packet delay and the buffer overflow ratios for all the routing algorithms as the arrival rate $\lambda_1$ is varied. We observe that D-ORCD has better delay performance than the other algorithms over the range of incoming traffic rates considered. 
Fig.~\ref{fig:QueueNextHop} plots the highest priority next hop for node $1$ under the candidate protocols throughout
the duration of the experiments. 

ExOR gives higher priority to node 2 than node 3 independent of the congestion at intermediate nodes 
($ETX^{(2,7)}=2.53$ and $ETX^{(3,7)}=4.36$). 
ExOR can thus suffer from poor delay performance  as the arrival rate at node 2 approaches capacity. 
  ExOR has the worst delay performance among all the algorithms as seem in Fig~\ref{fig:QueueNextHop} particularly when the traffic load on the network is high.
  In Fig.~\ref{fig:QueueNextHop} we observe that \mbox{DIVBAR} and \mbox{E-DIVBAR} forward significant number of packets into 12,13 and 14  increasing the interference and packet drops as well as delay. 

  Next, we study the impact of  the size of the ``hole''; i.e.  $N$  on the expected per packet delay. 
 Under DIVBAR the packets that arrive at node $2$ from source $1$ are likely to be forwarded and wander between nodes $12, 13 \ldots, 12+(N-1)$ before  eventually forwarding to $4$. 
In contrast, increasing $N$ has no effect on the performance of D-ORCD.
This is because $V_2^7(t) < V_{12+i}^7(t)$, $i=0,2,\ldots,N-1$, for all time slots $t$, in effect, preventing the packets to enter the ``hole''. 
%Note that \mbox{E-DIVBAR}, in this scenario, has a poorer performance,
%as the choice of the total sum of ETX and the backlog can create a further bias in routing the packets through node $2$. 
Fig.~\ref{fig:perfM2L10} provides the expected delay encountered by the source packets under various routing policies, as the size of the ``hole'', $N$, increases and the arrival rate is set to low value of $\lambda_1 = 200$ kbps. The figure shows that the average delay under D-ORCD is significantly lower than other candidate protocols as $N$ increases from $1$ to $5$.

\begin{figure}[!h]
\centering
\includegraphics[width=0.5\textwidth]{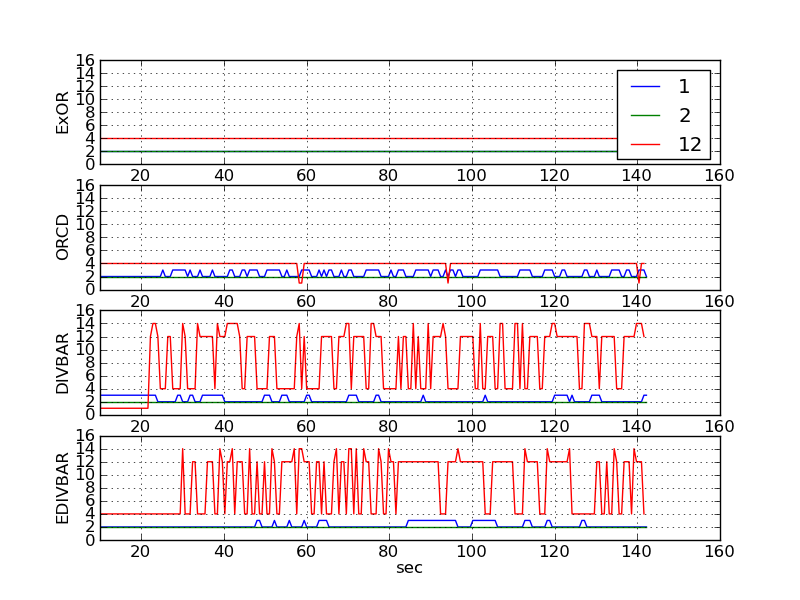}               
\caption{Highest priority nodes for Canonical Example.}
 \label{fig:QueueNextHop}
\end{figure}

\begin{figure}[ht]
\centering
%\subfigure[Delay] {
\includegraphics[width=.45\textwidth]{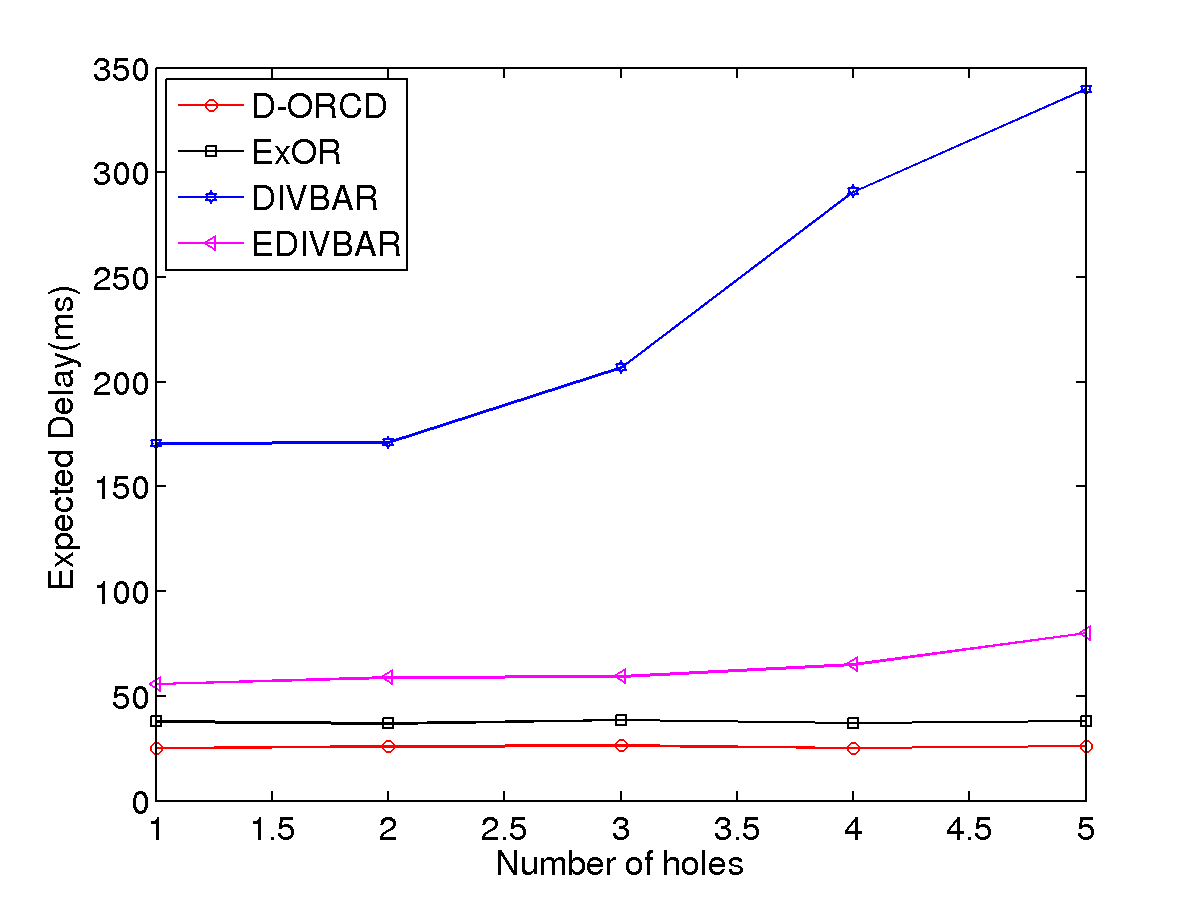}%
% }
%\subfigure[Fraction of packet loss]{
%%\includegraphics[width=0.45\textwidth]{figures/delivery_loss_grid_N.png}               
%%}
\caption{Performance for Canonical Example for  $\lambda$=200 kbps}
\label{fig:perfM2L10}
\end{figure}

\subsection{Performance of D-ORCD: Grid Topology}

\begin{figure}[ht]
\centering
\subfigure[Grid topoloogy. All nodes have the same arrival rate.] {
\includegraphics[width=.3\textwidth]{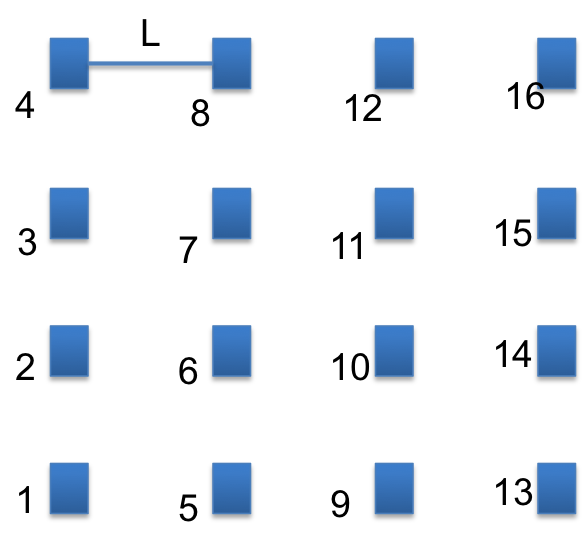}%
\label{grid1}%
}
\subfigure[Modifications to grid topology with blockage.All nodes have the same arrival rate, except node 10 does not generate traffic.] {
\includegraphics[width=.3\textwidth]{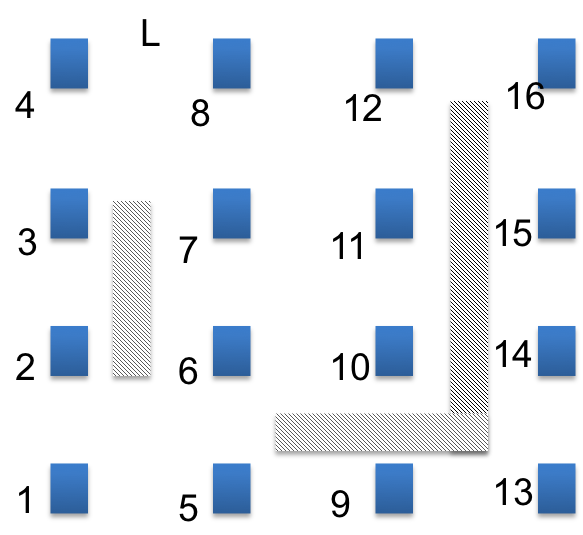}%
\label{grid2}
}
\caption{Grid topology of 16 nodes (4 x 4). Node 1 is assumed to be the destination}
\end{figure}

We perform simulations for the grid networks of 16 nodes in Fig. \ref{grid1} and \ref{grid2}. UDP Traffic is injected at each node $i \in \Omega$, with poison distributed packet arrivals.  Figure \ref{fig:grid_no_block} shows the expected delay versus the arrival rate  under various routing policies for the network in Fig. \ref{grid1}. Under ExOR, packets are always routed opportunistically along the ``shortest path'' to the destination which results in high delay under heavy traffic scenarios. On the other hand, \mbox{DIVBAR}, \mbox{E-DIVBAR}, and \mbox{D-ORCD} are throughput optimal and hence, they distribute the traffic to ensure bounded average delay for all traffic rates inside the stability region. The performance gap between DIVBAR and E-DIVBAR follows from the fact that DIVBAR does not use any metric of closeness to the destination when routing the packets; while E-DIVBAR takes into account the ETX of the nodes. A more interesting observation is the comparable performance of D-ORCD and E-DIVBAR. 
In other words, in network \ref{grid1}  the mere addition of ETX and queue measures in E-DIVBAR perform sufficiently well. 

\begin{figure}[ht]
\centering
\subfigure [Average delay per packet delivery for Network shown in Fig.\ref{grid1}]{
\includegraphics[width=.47\textwidth]{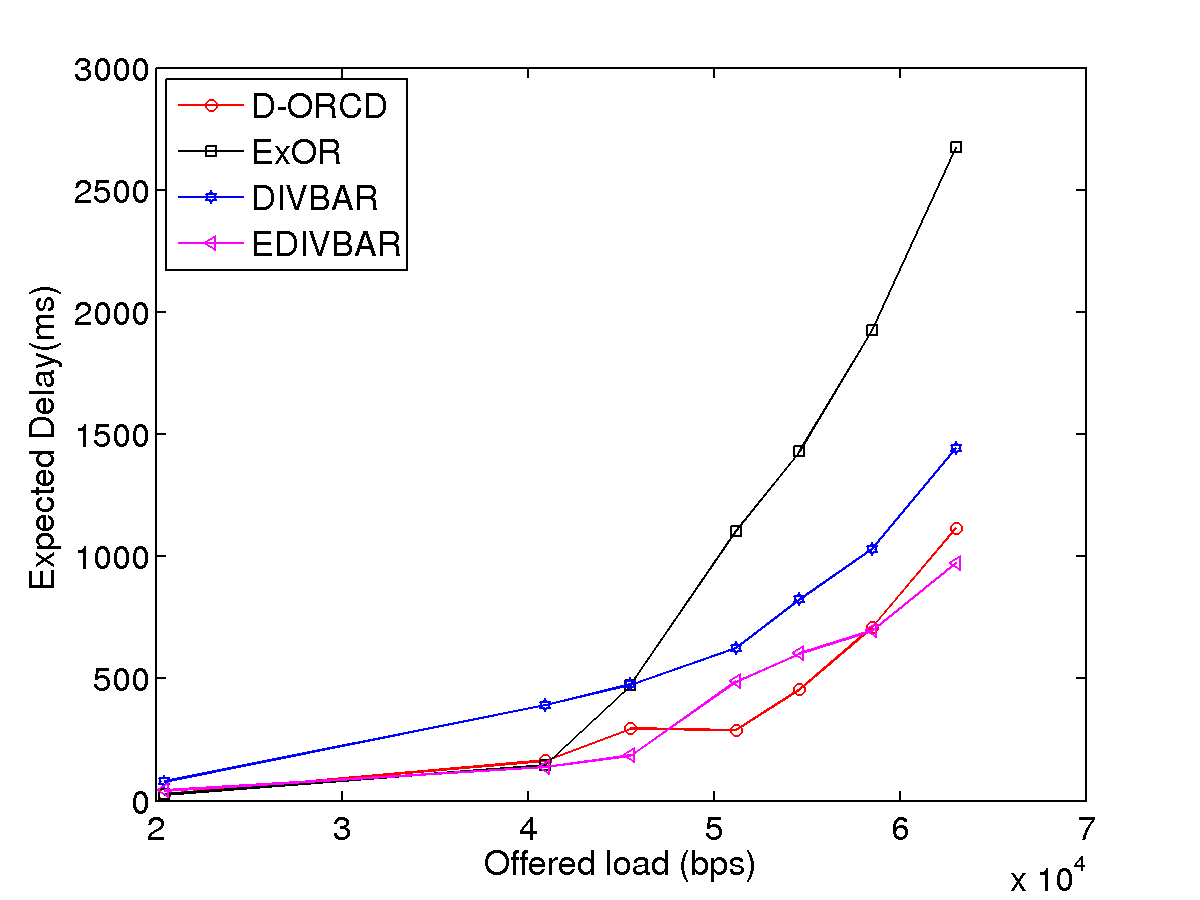}%
\label{fig:grid_no_block}
}
\subfigure [Fraction of the packets lost is dominated by FO packet loss. (Packet loss due to buffer overflow is negligible)] {
\includegraphics[width=.47\textwidth]{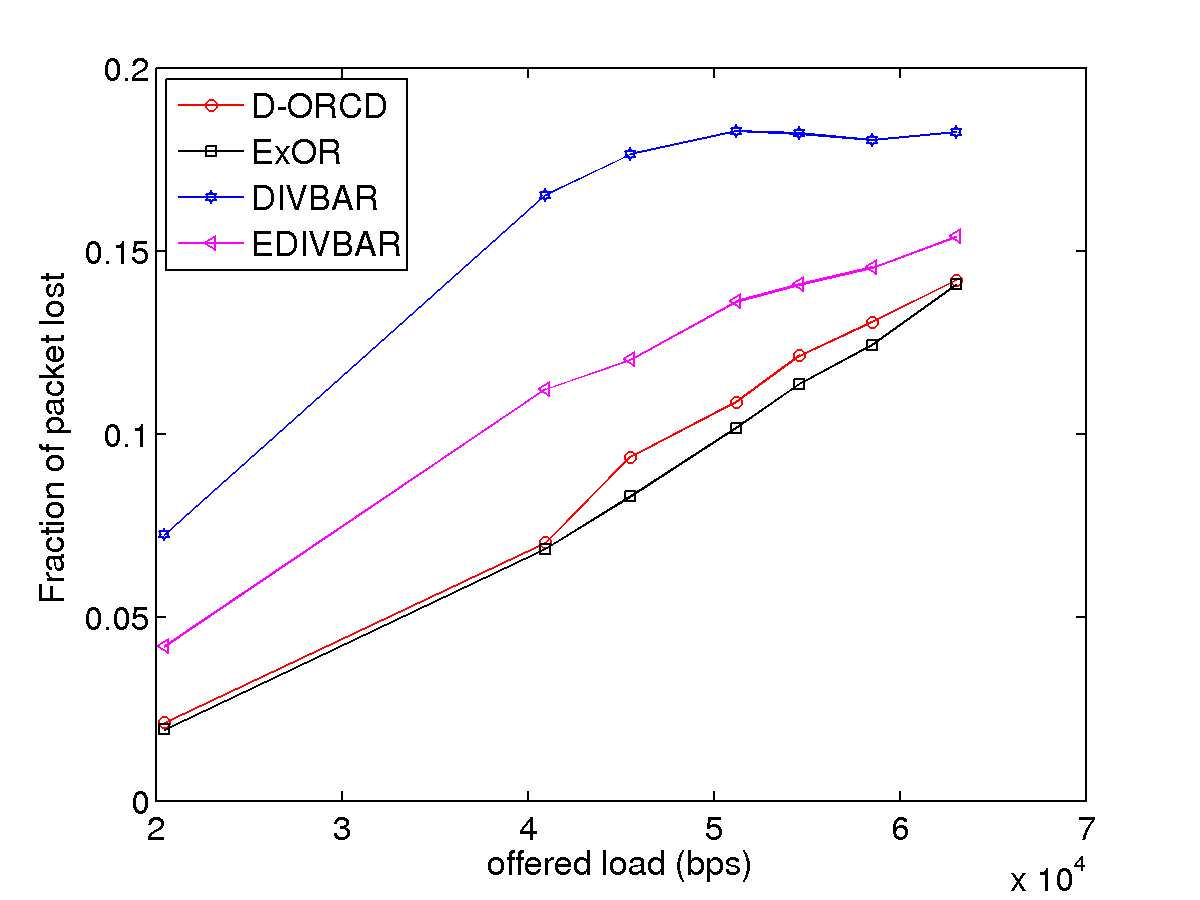}
%\caption{Fraction of the packets lost due to FO packet loss. Packet loss due to buffer overflow is negligible, hence not shown separately} \label{loss}
}
\caption{Performance results for the grid topology.}
 \end{figure}

Next, we consider the network shown in Fig.\ref{grid2}, a modification of the network shown in Fig.\ref{grid1} in which link qualities are changed due to the existing barriers in the network. Figure \ref{fig:grid_block} shows the delay performance of the candidate routing policies for this network as the traffic load varies. Again, ExOR and DIVBAR show large delay. But, unlike in the case of network shown in  Fig.\ref{grid1}, the performance gap between D-ORCD and E-DIVBAR is now rather significant. The reason is that D-ORCD always route packets along the least congested paths to the destination (without assuming the network topology and the arrival traffic). In other words, the performance of E-DIVBAR exhibits high dependence on the underlying network topology and the arrival traffic: E-DIVBAR performs well in symmetric networks with equal arrival rate to all nodes (e.g. the network of Fig.\ref{grid1}), while, it performs poorly in non-symmetric networks under non-uniform traffic patterns.

\begin{figure}[th]
\centering
\subfigure [Average delay per packet delivery for Network shown in Fig.\ref{grid2}.] {
\includegraphics[width=.47\textwidth]{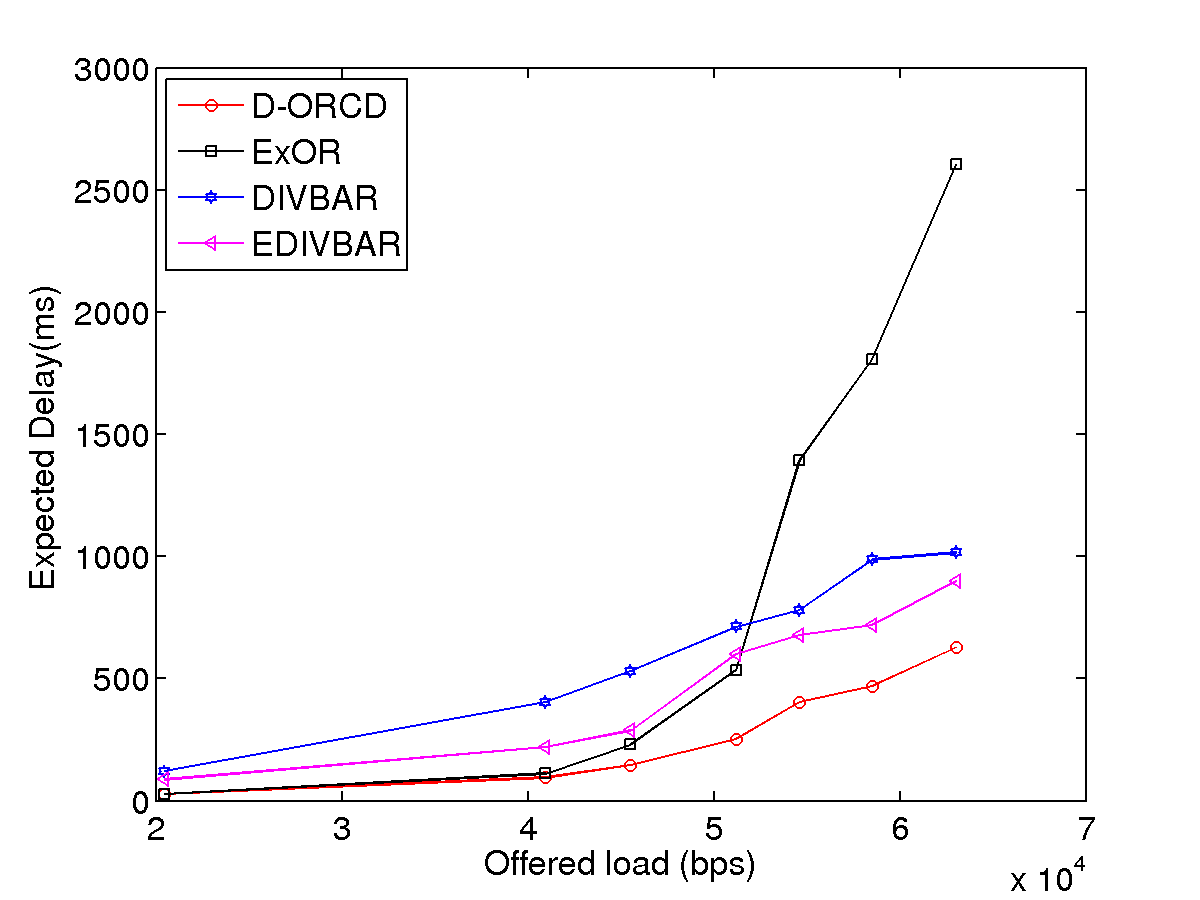}%
%\caption{Average delay per packet delivery for Network shown in Fig.\ref{grid2}.}
\label{fig:grid_block}
}
\subfigure [Fraction of the packets lost is dominated by FO packet loss. 
   (Packet loss due to buffer overflow is negligible)] { 
\includegraphics[width=.47\textwidth]{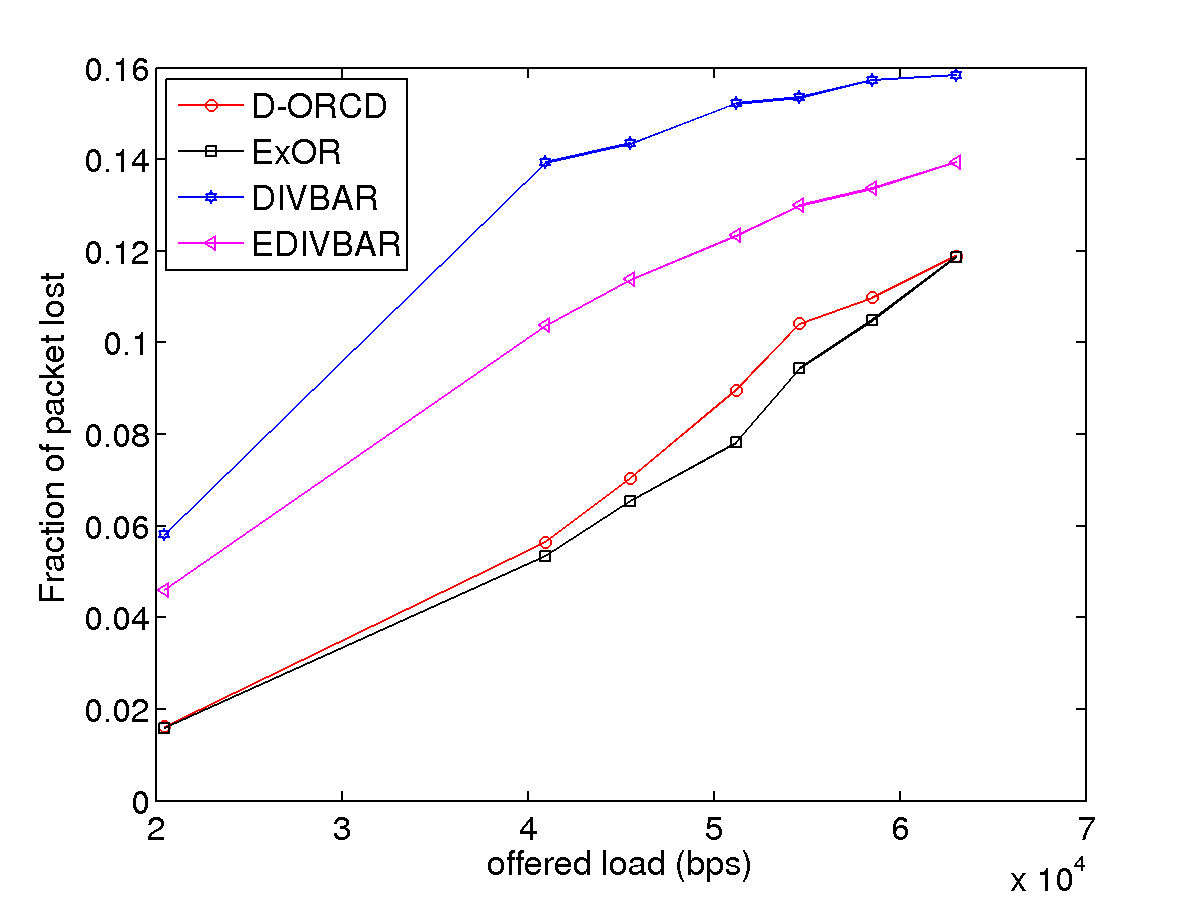}%
 \label{blockloss}
 }
 \end{figure}
 
 \begin{comment}
\clearpage
 \subsection{Urban Environment}
\begin{figure}[!h]
\centering
\includegraphics[width=.4\textwidth]{figures/los_angeles.png}%
\caption{Urban  Topology. 16 nodes placed on the grid in the environment of los angeles (34,-119.5 - 34.5,--119.1)}%
\label{grid3}%
\end{figure} 

\begin{figure}[!h]
\centering
\includegraphics[width=.4\textwidth]{figures/delivery_delay1_single.pdf}%
\caption{Average delay per packet delivery for Network shown in Fig.\ref{grid3}.}
\label{fig:urban}
\end{figure} 

\subsection{Interference and congestion diversity}

\begin{figure}[!h]
\centering
\includegraphics[width=.3\textwidth]{figures/diamand.png}%
\caption{Example Topology}%
\label{diamand}%
\end{figure}
Not a clean plot. Seems like ExOR also distribute the traffic and gets hurts equally.
This is not the same as we have seen in D-ORCD.  

\begin{figure}[!h]
\centering
\subfigure[Delay] {
\includegraphics[width=.45\textwidth]{figures/delivery_delay1_diamand.pdf}%
 }
\subfigure[Throughput]{
\includegraphics[width=0.45\textwidth]{figures/delivery_thru1_diamand.pdf}               
}
\end{figure}
\end{comment} 

\subsection{Choice of Parameters}
\label{choice}
Next, we investigate the performance of D-ORCD with respect to the design parameters in the grid topology of 16 nodes in Figure \ref{grid1}.  
It provides significant insight in the appropriate choice of the design parameters such as choice of partial diversity $M$ and choice of computation cycle $T$.  
\subsubsection{Choice of partial diversity $M$}
We focus on characterizing the trade-off between performance and overhead cost for D-ORCD. 
We consider modifications of D-ORCD with partial diversity to decide on the number of neighbors $M$ which acknowledge the reception of the packet. In particular, we compare the delay performance as well as the overhead cost of  D-ORCD. Figure \ref{fig:partialD} shows the average delivery time of each packet versus the number of $M$ for Network shown in Fig.\ref{grid1}. Figure \ref{fig:partialD} illustrates the trade-off between the delay performance and overhead cost D-ORCD. We note that limiting the size of the neighbor set to 4 provides the best trade-off. 

\begin{figure}[!h]
\centering
 {
\includegraphics[width=.5\textwidth]{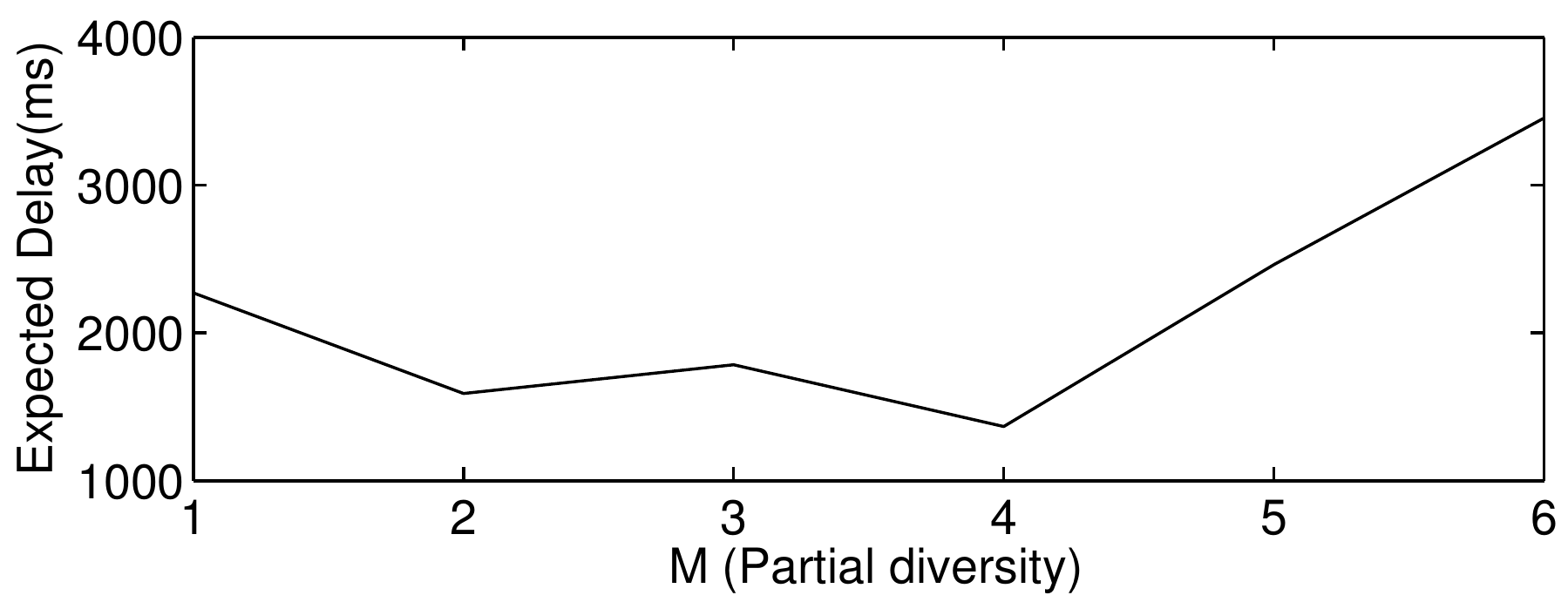}%
 }
\caption{Delay performance for D-ORCD  with partial diversity for Network shown in Fig.\ref{grid1}}
 \label{fig:partialD}
 \end{figure}

\subsubsection{Choice of computation cycle interval $T_c$}
\label{perfT}
D-ORCD throughput optimality as we will discuss in Section \ref{optimality} requires that computation cycle interval to be sufficiently large. However, 
to ensure a better delay performance, $T_c$ must be chosen sufficiently low to make the routing 
decisions more responsive to the instantaneous congestion. In particular, as $T_c$ increases, the chosen routing paths i) utilizes outdated queue lengths and ii) keeps the routing policy fixed for longer durations independent of current queue-lengths. In Figure \ref{fig:perfT}, we plot the performance 
of D-ORCD as $T_c$ varies in terms of multiple of $T_s$.  We observe that for high load, the choice $T_c=T_s$ outperforms other values for $T_c$. We have chosen a more responsive version of (\ref{fixpointDec_d}) at the cost of provable throughput optimality,  where $T_c$ is set to $T_s=0.5$ seconds.

\begin{figure}[ht]
\centering
\includegraphics[width=.45\textwidth]{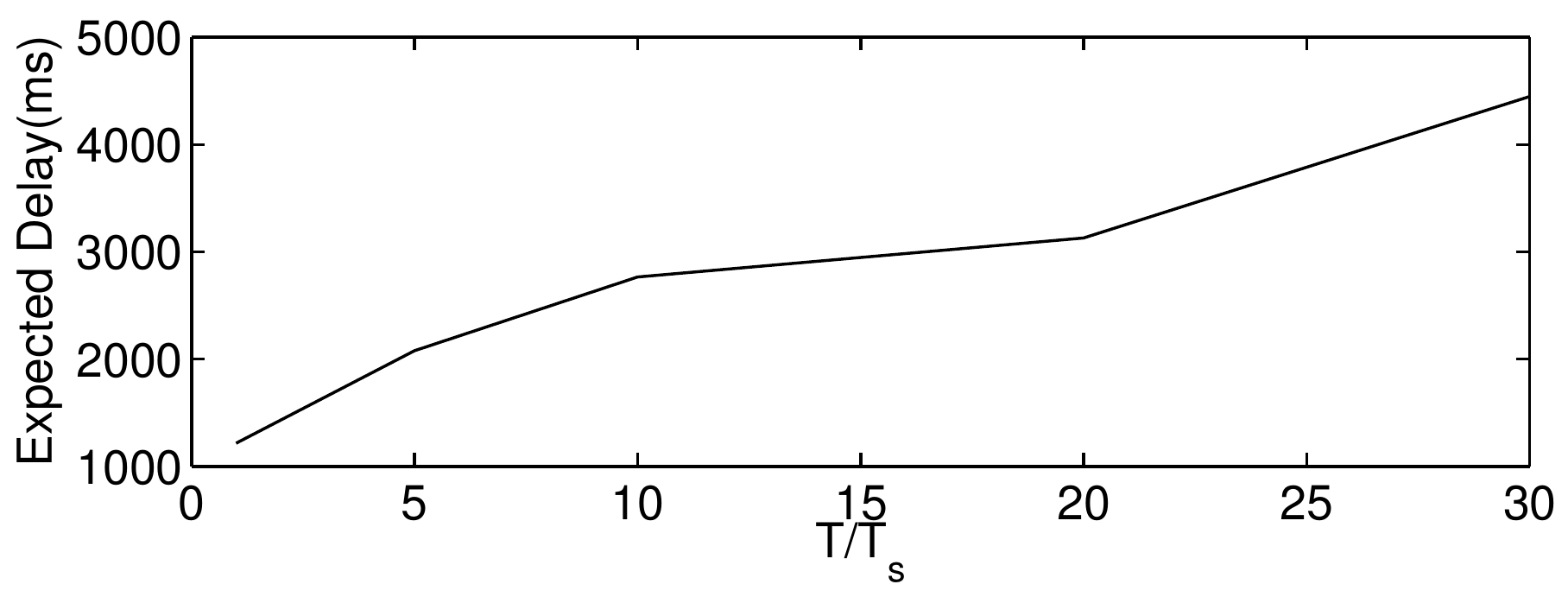}%
\caption{Delay Performance for D-ORCD for Network shown in Fig.\ref{grid1} as $T_c$ varies}
\label{fig:perfT}
\end{figure}

\section{Theoretical Guarantees}
\label{optimality}
 In this section, we provide a theoretical guarantee regarding the throughput optimality of  D-ORCD under the assumptions 
that i)the flows in the network are destined for the single destination node D (for multi-destination extensions see \cite{Hairo}), ii) link probabilities are time invariant iii) the routing decisions and the successful reception at set $S$ due to transmission from node $i$ is acknowledged perfectly to node $i$.

Before we precisely state the optimality, we define few notations.
We define a \emph{routing decision} $\mu_{ij}(t)$ to be the number of packets (upto 1 packet) whose relaying responsibility is shifted from node $i$ to node $j$ during time slot $t$ ($\mu_{ii}(t)=1$ means that $i$ retains the packet). 
Note that $\mu_{ij}(t)$ forms the departure process from node $i$, while it 
creates an endogenous arrival to node $j$. 
Without loss of optimality, we assume that $p_{ii}=1$ and $\mu_{iD}(t)=1$, if $D \in S_i(t)$.
\begin{definition}
A \emph{routing policy} is a collection of causal routing decisions $\cup_{i,j \in \Omega} \cup_{t=0}^{\infty} \{ \mu_{ij}(t) \}$.\end{definition}

Let $A_i(t)$ represent the amount of data that exogenously arrives to node $i$ during time slot $t$. Arrivals are assumed to be i.i.d. over time and bounded by a constant $A_{max}$. Let $\lambda_i = \mathbb{E} [A_i(t)]$ denote the exogenous arrival rate to node $i$. We define $\boldsymbol{\lambda}=[\lambda_1, \lambda_2, \ldots, \lambda_N]$ to be the arrival rate vector. 
Let $Q_i(t)$ denote the queue backlog of node $i$ at time slot $t$. 
We assume any data that is successfully delivered to the destination $D$ will exit the network and hence, $Q_D(t) = 0$ for all time slots $t$.
We define $\boldsymbol{Q}(t)=[Q_1(t), Q_2(t), \ldots, Q_{D-1}(t)]$ to be the vector of queue backlogs of nodes $1,2,\ldots,D-1$.

The selection of routing decisions under a routing policy $\Pi$ together with the exogenous arrivals impact the queue backlog of node $i$, $i \in \Omega$ as:
\begin{eqnarray}
\label{Qdynamic}
\nonumber
Q^{\Pi}_i(t+1) & =& [ Q^{\Pi}_i(t) - \sum_{j \in \Omega} \mu^{\Pi}_{ij} (t) ] ^{+}  \\
\nonumber
&& + \sum_{j \in \Omega} \mu^{\Pi}_{ji}(t) \mathbf{1}_{\{Q^{\Pi}_j(t) \ge \mu^{\Pi}_{ji}(t)\}} + A_i(t),
\end{eqnarray}
where the superscript $\Pi$ emphasizes the dependence of queue backlog dynamics on the
choice of policy $\Pi$. 

\begin{definition}
Given an ergodic exogenous arrival process with rate $\boldsymbol{\lambda}$, a routing policy $\Pi$ is said to \emph{stabilize} the network if $Q^{\Pi}_{tot} (t)$ is ergodic and $\mathbb{E} [Q^{\Pi}_{tot} (t)]$ remains bounded when packets are routed according to $\Pi$. 
The \emph{stability region} of the network (denoted by $\mathfrak{S}$) is the set of all arrival rate vectors $\boldsymbol{\lambda}$ for which there exists a routing policy that stabilizes the network. 
\end{definition}

\begin{definition}
A routing policy is said to be \emph{throughput optimal} if it stabilizes the network for all arrival rate vectors that belong to the interior of the stability region.
\end{definition}

\begin{fact}[Corollary 1 in \cite{Neely09}]
\label{Neelyeps}
%(\cite[]{Neely09}) 
%Let $\mathfrak{S}$ denote the stability region of the network. 
An arrival rate vector $\boldsymbol{\lambda}$ is within the stability region $\mathfrak{S}$ if and only if there exists a stationary randomized routing policy that makes routing decisions $\{ \tilde{\mu}_{ij}(t) \}_{i,j \in \Omega}$, solely based on the collection of potential forwarders at time $t$, $\{ S_i(t) \}_{i \in \Omega}$, and for which 
$$\mathbb{E} \left[ \sum_{j \in \Omega} \tilde{\mu}_{kj}(t) - \sum_{i \in \Omega} \tilde{\mu}_{ik}(t) \right] \ge \lambda_k.$$
%where the expectation is taken with respect to the random topology state and the (potentially) random control action based on this state.
\end{fact}

We are ready to present Theorem~\ref{SD-ORCDopt} regarding the optimality of \mbox{D-ORCD}.   
\begin{theorem}
\label{SD-ORCDopt}
Suppose $T_c=O(D)$ and $M=D$. Then D-ORCD is throughput optimal.
\end{theorem}
The proof of Theorem~\ref{SD-ORCDopt} is based on the Foster-Lyapunov Theorem. 
For completeness, the structure of the Lyapunov function and a sketch of the proof is provided in the Appendix.
By Theorem~\ref{SD-ORCDopt}, under D-ORCD, the average total queue backlog remains bounded. Little's theorem implies that under D-ORCD, expected delay is bounded.

%\section{Discussion of the Model}
%\label{DiscussModel}

\begin{remark}
Assumptions for optimality of D-ORCD could be relaxed in many cases. 
\begin{enumerate}
\item 
The packet transmission on a link $(i, j)$ is assumed to be successful with probability $p_{ij}$, and transmissions on links were assumed to be independent of each other. The computations in (\ref{fixpointDec_d}) and (\ref{infrequent}) can be generalized to incorporate correlated link qualities. 
by replacing the term {\small{
$\big( \prod_{k \in S} p_{ik} \big) \big( \prod_{l \notin S} (1-p_{il}) \big)$}} with $P(S|i)$ in the definition of $P^{(i,d)}(t)$ and $P^{(i,d)}_{succ-k}(t)$, where $P(S|i)$ denotes be the probability of the event $\{S_i(t) = S\}$.
Furthermore, it is straight forward to show that the throughput optimality of D-ORCD is robust to all channel estimation errors, even though, erroneous link models, in general, can significantly degrade its delay performance.
\item 
In this paper, we assumed that the network topology and the probability of successful transmissions are time-invariant.
The generalization to the case of time-varying network topology with stationary transmission probabilities is straight forward \cite{Mohammad10}.
%\item 
%Optimality of D-ORCD assumes that ACK and FO packets are not subject to error or delay. 
%However, this could be relaxed 
\end{enumerate}
\end{remark}

\section{Conclusions and Discussions}
\label{Conclusion}

In this paper, combining the important aspects of shortest path routing with those of backpressure routing, we provided a distributed  opportunistic
routing policy with congestion diversity (D-ORCD) is proposed under which packets are routed according to a rank ordering of the nodes based on a congestion cost measure. 
  Furthermore, we show that D-ORCD allows for a practical distributed and asynchronous 802.11 compatible implementation, whose performance was investigated via a detailed set of QualNet simulations under practical and realistic networks. Simulations show that D-ORCD consistently outperforms existing routing algorithms in practical settings.

In D-ORCD, we do not model the interference from the nodes in the network, but instead leave that issue to a classical MAC operation. 
However, the generalization to the networks with inter-channel interference follows directly from \cite{Neely09}. The price of this generalization is shown to be the centralization of the routing/scheduling globally across the network or a constant factor performance loss of the distributed variants \cite{Neely09,Xi06,Ying08}.
In future, we are interested in generalising D-ORCD for joint routing and scheduling optimizations as well consider system level implications. 
Incorporating throughput optimal CSMA based MAC scheduler (proposed in \cite{Walrand}) with congestion aware routing is also promising area of research.

The design of D-ORCD requires knowledge of channel statistics. Designing congestion control routing algorithms to minimize expected delay without the topology and the channel statistics knowledge is an area of future research. 

 \section*{Acknowledgements}
 This work was partially supported by the industrial sponsors of UCSD Center for Wireless Communications (CWC) and Center for Networked Systems (CNS) and UC Discovery Grant \#com07-10241. 
 The authors would like to thank Mr. Anders
 Plymoth for valuable discussions and suggestions.
 
 \appendix
We provide a sketch of the proof for the throughput optimality of D-ORCD for a connected network.\footnote{In connected network each node has a positive probability path to the destination. If a node has no path to the destination, it cannot sustain any traffic and can be ignored without loss of generality.} 

\subsection{Relationship to Centralized ORCD}
We prove the throughput optimality by relating 
 D-ORCD update equation (\ref{fixpointDec_d}) to the convergence of closely related fixed point equation. 
  In particular, we relate the routing decisions for D-ORCD with the decisions taken according to the congestion measures $\{V_i^*(t)\}$ obtained from the fixed point equation:  
\begin{equation}
\label{infrequent}
V_i^*(t) = \frac{Q_i(t)}{P^{(i,D)}(t)} + \sum_{k : H^{(i,D)}(t)}  \frac{P^{(i,D)}_{succ-k}(t)}{P^{(i,D)}(t)} V_k^*(t). 
\end{equation}	

We refer to the centralized routing algorithm which makes decisions at each instant according to $V_i^*(t)$ as {{\mbox{C-ORCD}}}.  
Following lemma states a relationship between D-ORCD and C-ORCD. 
\begin{lemma}
\label{converge}
Assume $T_c$ is sufficiently large $(T_c \sim O(D))$. Then during  $ T(t) \leq t <  T(t+T_c) $, (\ref{fixpointDec_d}) converges
 to the fixed point equation (\ref{infrequent}), 
i.e. $\{{V}_i^D(T(t+T_c)\}_{i \in \Omega}$ solves (\ref{infrequent}) and  ${V}_i^D(T(t+T_c)) = V_i^*(T(t))$.
\end{lemma}
 \begin{proof}
 The convergence ${V}_i^D(t) \to V_i^*(T(t))$ during $T(t) \leq t < T(t+T_c) $ follows
by relating (\ref{fixpointDec_d}) to  the  Bellman-Ford algorithm with fixed link cost. It is known   
from \cite[Theorem 2.4]{Tsitsiklis95} that asynchronous distributed Bellman-Ford algorithm converges in finite time when the control packets are instantaneously received (or control packets are timestamped and older packets are discarded). Furthermore, with high probability, the time until the termination of the asynchronous Bellman-Ford algorithm is	$\mathbf{O}(D)$, where $D$ is the number of nodes in the network.
\end{proof}

Note that when implementing D-ORCD, we broadcast control packets using high priority and the control packets do not undergo backoff. This ensures that with high 
probability the packets are instantaneously received. Thus the convergence in  Lemma \ref{converge} is justified.   

%We next argue that deviating from C-ORCD decisions for finite duration does not affect throughput optimality and thus D-ORCD is also throughput optimal.

We now provide the proof for the throughput optimality of D-ORCD for $T_c\sim O(D)$. 
To simplify the notations, let policies in C-ORCD and D-ORCD be denoted by $\pi^*$ and $\hat{\pi}$ respectively.
In \cite{Naghshvar09}, the authors constructed an appropriate Lyapunov function $L^*$ to show that  C-ORCD is throughput optimal. In particular it is shown that:  
\begin{fact}
\label{LyapStable}
There exists a  Lyapunov function. $L^*: {\mathbb{R}}^{D}_{+} \to \mathbb{R}_+$  such that for all time slots $t$ and $B>0$, $\epsilon>0$, 
\begin{align}
\label{LyapStableeqn}
\mathbb{E} \left[ L^*_{}(\boldsymbol{Q}^{\pi^*} (t+1)) - L^*_{}(\boldsymbol{Q}(t)) | \boldsymbol{Q}(t) \right] \le B - \epsilon \sum_{k=1}^{N} Q_k(t),
\end{align}
where superscript $\pi^*$ in $\boldsymbol{Q}^{\pi^*}$ implies the dependence of backlog vector on the routing policy $\pi^*$.  
\end{fact}

To prove throughput optimality of D-ORCD, it suffices to show that 
the Lyapunov drifts under $\pi^*$ and $\hat{\pi}$ have a bounded difference. More precisely, 
\begin{lemma}
\label{BDiff}
Let $L^*$ be the Lyapunov function as proposed in \cite{Naghshvar09}.
Then for $B''>0$, 
\begin{align*}
\mathbb{E} \left[ L^{*}(\boldsymbol{Q}^{\hat{\pi}}(t+1)) | \boldsymbol{Q}(t) \right] 
- \mathbb{E} \left[ L^{*}(\boldsymbol{Q}^{{\pi^*}}(t+1)) | \boldsymbol{Q}(t) \right] < B''.
\end{align*}
\end{lemma}

Lemma \ref{BDiff} together with (\ref{LyapStableeqn}) implies the existence of Lyapunov function with negative 
expected drift i.e. for  Lyapunov function $L^*$, there exists $B'>0$, $\epsilon'>0$ such that, 
\begin{align}
\nonumber
\mathbb{E} \left[ L^*_{}(\boldsymbol{Q}^{\hat{\pi}}(t+1)) - L^*_{}(\boldsymbol{Q}^{\hat{\pi}}(t)) | \boldsymbol{Q}(t) \right] \le B' - \epsilon' \sum_{k=1}^{N} Q_k(t). 
\end{align}

The details of the construction of $L^*$ and the proof of Lemma~\ref{BDiff} is provided in Appendix~\ref{app:drift}.

%%%%%%%%%%%%%%%%%%%%%%%%%%%%%%%%%%%%%%%%%%%%%%%%%%%%%%%%%%
%%%%%%%%%%%%%%%%%%%%%%%%%%%%%%%%%%%%%%%%%%%%%%%%%%%%%%%%%%
\subsection{Review of C-ORCD results}

Before proceeding, we introduce some notations. Let $[ x ] ^{+} = \max \{x,0 \}$. The indicator function $\mathbf{1}_{\{ X \}}$ takes the value $1$ whenever event $X$ occurs, and $0$ otherwise. For any set $S$, $\left| S \right|$ denotes the cardinality of $S$, while for any vector $\boldsymbol{v}$, $\left\| \boldsymbol{v} \right\|$ denotes the euclidean norm of $\boldsymbol{v}$.
%For any set $S$, $int(S)$ is the set of all interior points of $S$.
When dealing with a sequence of sets $C_1,C_2,\dots$, we define $C^i=\cup_{j=1}^{i} C_j$. 

%%%%%%%%%%%%%%%%%%%%%%%%%%%%%%%%%%%%%%%%%%%%%%%%%%%%%
%\subsubsection{Construction of the Lyapunov Function}

The following definitions are required in order to identify the Lyapunov functions for C-ORCD and D-ORCD.

\begin{definition}
A \emph{rank ordering} $R=(C_1,C_2,\ldots,C_M)$ is an ordered list of non-empty sets $C_1,C_2,\ldots,C_M$ $(1\le M \le D)$, referred to as \emph{ranking classes}, that make up a partition of the set of nodes $\{ 1, 2, \ldots, D \}$, i.e., $\cup_{i=1}^{M} C_i = \{ 1,2,\ldots,D \}$ and $C_i \cap C_j = \emptyset$, $i \neq j$. 
\end{definition}

\begin{definition}
A rank ordering $R=(C_1,C_2,\ldots,C_M)$ is referred to as \emph{path-connected} if for each node $i \in C_k$, $1 \le k \le M$, there exist distinct nodes  $j_1,j_2,\ldots,j_l \in C^{k-1}$ such that $p_{ij_1}>0, p_{j_1 j_2}>0, \ldots, p_{j_l D}>0$. 
The set of all path-connected rank orderings is denoted by $\mathcal{R}^{c}$.
\end{definition}

Let $f$ be a bivariate function of the following form:
\begin{align*}
%\label{Definef}
f(m,n)=\frac{1}{K^m (K^n -1)} \ \ \text{for all} \ m\ge0,n>0,
\end{align*}
where $K = 1+\frac{1}{p_{\text{min}}}$ for $p_{\text{min}}=\min \{ p_{ij}: i , j \in \Omega, p_{ij}>0 \}$.

In \cite{Naghshvar09}, the authors proposed a method that utilizes the bivariate function $f$ and partitions the space of queue backlogs, $\mathbb{R}^D_+$, 
into $|\mathcal{R}^c|$ cones denoted by $\{ D^c_f(R) \}_{R \in \mathcal{R}_c}$.
The piece-wise Lyapunov function, $L^{*}_{f}: \mathbb{R}^{D}_{+} \to \mathbb{R}_{+}$, is then constructed by assigning to each cone 
$D^c_f(R)$, $R=(C_1,C_2,\ldots,C_M) \in \mathcal{R}_c$, a weighted quadratic function of the form:
\vspace{-0.2cm}
\begin{align}
L_{f}(\boldsymbol{Q},R)= \sum_{i=1}^M f(|C^{i-1}|,|C_{i}|) Q^2_{C_i},
\end{align}
where $Q_{C}(t)= \sum \limits_{i \in C} Q_i(t)$. 
\vspace{-0.2cm}
%Since the collection of cones form a partition of $\mathbb{R}^{D}_+$, we can combine the above quadratic functions to arrive at a piece-wise quadratic function
More precisely,
\begin{align}
\label{Lyapunov}
L^{*}_{f}(\boldsymbol{Q})= \sum_{R \in \mathcal{R}_{c}} L_{f}(\boldsymbol{Q},R) \mathbf{1}_{\{\boldsymbol{Q} \in D^c_f(R)\}}. 
\end{align}  

%%%%%%%%%%%%%%%%%%%%%%%%%%%%%%%%%%%%%%%%%%
%\subsubsection{Proof of Claim~\ref{BDiff}}

%We now provide the proof of Claim~\ref{BDiff}.

Let us consider the Lyapunov drift when $\boldsymbol{Q}(t) \in D^c_f(R)$ for some $R=(C_1,C_2, \ldots, C_M) \in \mathcal{R}_c$.
Since the collection of cones $\{ D^c_f(R) \}_{R \in \mathcal{R}_c}$ partitions $\mathbb{R}^{D}_+$, we define function 
$U_f: \Omega \times \mathbb{R}^{D}_+ \to \mathbb{R}_+$ such that  
\begin{align}
\label{DefineV}
U_f(k, \boldsymbol{Q}) = f(|C^{i-1}|,|C_i|) Q_{C_i}(t),
\end{align}    
where $\boldsymbol{Q} \in D^c_f(R)$, $R=(C_1,C_2,\ldots,C_M)$, and $k \in C_i$.

Let $\{\mu^*_{ij} (t)\}_{i,j \in \Omega}$ represent the routing decisions under C-ORCD, while 
$\{\hat{\mu}_{ij} (t)\}_{i,j \in \Omega}$ represent the routing decisions under D-ORCD.
  For ease of notation and exposition define $A_{C}(t)= \sum \limits_{i \in C} A_i(t)$, 
$\mu_{C,in}(t)=\sum_{k=1}^D \sum_{j \in C} \mu_{kj} (t) {\bf{1}}_{\{Q_k(t)\ge 1\}}$, and  
$\mu_{C,out}(t)=\sum \limits_{k \in C} \sum_{j=1}^D \mu_{kj}(t) {\bf{1}}_{\{Q_k(t)\ge 1\}}$.
We have,

%By Lemma~\ref{contdiff}, $L^*_f(\cdot)$ is continuous and differentiable. Thus, we can write $L^*_f(\boldsymbol{Q}(t+1))$ in terms of its first-order Taylor expansion around $L^*_f(\boldsymbol{Q}(t))$ and we obtain 
\begin{align}
\nonumber
\lefteqn{L^{*}_{f}(\boldsymbol{Q}(t+1)) - L^{*}_{f}(\boldsymbol{Q}(t))} \\
\nonumber
 &\stackrel{(a)}{=} \sum_{i=1}^M f(|C^{i-1}|,|C_{i}|) \left[ Q^2_{C_i}(t+1) - Q^2_{C_i}(t) \right]  \\
 \nonumber
& \hspace{0.5in} + O(\left\| \boldsymbol{Q}(t+1) - \boldsymbol{Q}(t) \right\|^2)\\
\nonumber
&\stackrel{(b)}{=} - 2 \sum_{i=1}^M f(|C^{i-1}|,|C_{i}|) Q_{C_i}(t) \Big  ( \mu_{C_i,out}(t) - \mu_{C_i,in}(t) \\ 
\nonumber
& \hspace{0.5in}   - A_{C_i}(t) \Big ) + O(1) \\
\nonumber
& = - 2 \sum_{i=1}^M f(|C^{i-1}|,|C_{i}|) Q_{C_i}(t) {\bf{1}}_{\{Q_k(t)\ge 1\}} \\
\label{equal2}
& \hspace{0.5in}\Big (\sum_{j=1}^D \sum \limits_{k \in C_i} \mu_{kj}(t)- \sum_{k=1}^D  \sum \limits_{j \in C_i} \mu_{kj}(t) 
- A_{C_i}(t) \Big ) + O(1),
\end{align}
where $(a)$ follows from continuity and differentiability of $L^*_f$~\cite[Lemma~3]{Naghshvar09} 
and writing $L^*_f(\boldsymbol{Q}(t+1))$ in terms of its first-order Taylor expansion around $L^*_f(\boldsymbol{Q}(t))$, 
and $(b)$ follows from Fact~\ref{queuedynamic_ineq} below.

\begin{fact}[\cite{Naghshvar09}]
\label{queuedynamic_ineq}
Let $R=(C_1,C_2, \ldots, C_M) \in \mathcal{R}$ and $\boldsymbol{Q}(t) \in D_f(R)$. We have
\begin{align*}
\lefteqn{Q^2_{C_i}(t+1) - Q^2_{C_i}(t) =}\\
& \hspace{0.5in} \beta_f - 2 Q_{C_i}(t) (\mu_{C_i,out}(t) - \mu_{C_i,in}(t) - A_{C_i}(t)),
\end{align*}
where $\beta_f$ is a constant bounded real number.
\end{fact}

Finally from  (\ref{DefineV}) and (\ref{equal2}), 
\begin{align}
\nonumber
\lefteqn{L^{*}_{f}(\boldsymbol{Q}(t+1)) - L^{*}_{f}(\boldsymbol{Q}(t))} \\
\nonumber
& = - 2 \sum_{i=1}^M  \sum_{k\in C_i} \sum_{l=1}^D  U_f(k,Q) \mu_{kl} (t) {\bf{1}}_{\{Q_k(t)\ge 1\}} \\
\nonumber
& - \sum_{i=1}^M \sum_{l\in C_i} \sum_{k=1}^D U_f(l,Q) \mu_{kl} (t)  {\bf{1}}_{\{Q_k(t)\ge 1\}} - \sum_{i=1}^M A_{C_i}(t)  + O(1) \\
\nonumber
& = - 2 \sum_{k=1}^D \sum_{l=1}^D  U_f(k,Q) \mu_{kl} (t) {\bf{1}}_{\{Q_k(t)\ge 1\}} \\
\label{equality} 
& \hspace{0.2in} - \sum_{k=1}^D  \sum_{l=1}^D U_f(l,Q) \mu_{kl} (t)  {\bf{1}}_{\{Q_k(t)\ge 1\}} - \sum_{i=1}^M A_{C_i}(t)  + O(1). 
\end{align}

\begin{fact}[\cite{Naghshvar09}]
\label{Croute}
Routing decisions under C-ORCD are such that $\mu^*_{ij}=1$, only when $j \in S_i(t)$ and 
$U_f(j,\boldsymbol{Q}(t)) \le U_f(k,\boldsymbol{Q}(t))$ for all $k \in S_i(t)$.
\end{fact}  
This fact together with (\ref{equality}) provides the proof of Fact~\ref{LyapStable}. 
%-------------------------------------------------------------------------------------

\subsection{Proof of Lemma~\ref{BDiff}}
\label{app:drift}

We begin the proof of Lemma~\ref{BDiff} by stating following Claim.

\begin{claim}
\label{Droute}
Routing decisions under D-ORCD are such that $\hat{\mu}_{ij}=1$, only when $j \in S_i(t)$ and 
$U_f(j,\hat{\boldsymbol{Q}}(t)) \le U_f(k,\hat{\boldsymbol{Q}}(t))$ for all $k \in S_i(t)$,
where $\hat{\boldsymbol{Q}}(t)=\boldsymbol{\bar{Q}}(T(t))$. 
\end{claim}  

With this, we are ready to proceed with the proof of Lemma~\ref{BDiff}. 

\begin{align}
\label{PrfThm2_01}
\nonumber
\lefteqn{ \mathbb{E} [ L^{*}_{f}(\boldsymbol{Q}^{\hat{\pi}}(t+1))|\boldsymbol{Q}(t)] - \mathbb{E} [ L^{*}_{f}(\boldsymbol{Q}^{\pi^*}(t+1))|\boldsymbol{Q}(t)]} \\
\nonumber
&= \mathbb{E} \bigg[ 2 \sum_{k=1}^D \sum_{l=1}^D (\hat{\mu}_{kl}(t) - \mu^*_{kl}(t)) \Big(U_f(k,\boldsymbol{Q}(t)) \\ 
& \hspace{0.5in} - U_f(l,\boldsymbol{Q}(t) \Big) {\bf{1}}_{\{Q_k(t)\ge 1\}} | \boldsymbol{Q}(t) \bigg] + O(1), 
\end{align}
where equality follows from (\ref{equality}). 

Suppose node $i$'s transmission at time $t$ is received by potential forwarders $S_i(t)$.
Furthermore, suppose that nodes $a,b \in S_i(t)$ are the nodes with the highest rank under C-ORCD and D-ORCD respectively, i.e. $\mu^*_{ia}(t)=\hat{\mu}(t)_{ib}=1$. 
From Fact~\ref{Croute} and Claim~\ref{Droute}, we have
\begin{eqnarray}
\label{InfreqClaim_01}
U_f(a,\boldsymbol{Q}(t)) &\le& U_f(b, \boldsymbol{Q}(t)), \\ 
\label{InfreqClaim_02}
U_f(a,\hat{\boldsymbol{Q}}(t)) &\ge& U_f(b,\hat{\boldsymbol{Q}}(t)).
\end{eqnarray}
In order to prove Lemma~\ref{BDiff} it suffices to show that 
\begin{eqnarray}
\label{InfreqClaim_03}
U_f(b,\boldsymbol{Q}(t)) - U_f(a, \boldsymbol{Q}(t)) = O(\| \boldsymbol{Q}(t) - \hat{\boldsymbol{Q}}(t) \|).
\end{eqnarray}

Consider the line that connects $\boldsymbol{Q}(t)$ and $\hat{\boldsymbol{Q}}(t)$ in $\mathbb{R}^{D}_+$. Suppose this line goes through $M-1$ cones in $\mathbb{R}^{D}_+$. Let $Z_1,Z_2,\dots,Z_M$ be respectively the intersection of the line connecting $\boldsymbol{Q}(t)$ to $\hat{\boldsymbol{Q}}(t)$ with the $M$ separating hyperplanes of the $M-1$ cones between them, i.e.\
\begin{align}
%\label{InfreqClaim_04}
\nonumber
\lefteqn{\| \boldsymbol{Q}(t) - \hat{\boldsymbol{Q}}(t) \|} \\
\nonumber
&= \| \boldsymbol{Q}(t) - \boldsymbol{Z_1} \| + \| \boldsymbol{Z_1} - \boldsymbol{Z_2} \| + \cdots + \| \boldsymbol{Z_M} -  \hat{\boldsymbol{Q}}(t) \|.
\end{align}

Note that since $Z_1,Z_2,\dots,Z_M$ are on the hyperplanes, every two consecutive points in set $\{ \boldsymbol{Q}(t), Z_1,Z_2, \ldots, Z_M, \hat{\boldsymbol{Q}}(t) \}$ can be considered to belong to the same cone, and hence, have same rank ordering of the nodes. From definition of function $U_f$, we obtain
{\small{ \begin{align*}
      | U_f(a,\boldsymbol{Q}(t)) - U_f(a,\boldsymbol{Z_1})| &= O( \| \boldsymbol{Q}(t) - \boldsymbol{Z_1} \|),  \\
      | U_f(a,\boldsymbol{Z_m}) - U_f(a,\boldsymbol{Z_{m+1}})| &= O( \| \boldsymbol{Z_m} - \boldsymbol{Z_{m+1}} \|), \ 1 \le m \le M, \\
      | U_f(a,\boldsymbol{Z_M}) - U_f(a,\hat{\boldsymbol{Q}}(t))| &= O( \| \boldsymbol{Z_M} - \hat{\boldsymbol{Q}}(t) \|).                   
\end{align*} }} 
Therefore,
\begin{align}
%\label{InfreqClaim_05}
\nonumber
\lefteqn{U_f(a,\boldsymbol{Q}(t)) - U_f(a,\hat{\boldsymbol{Q}}(t))} \\
\nonumber
&=  [U_f(a,\boldsymbol{Q}(t)) - U_f(a,\boldsymbol{Z_1})] + [U_f(a,\boldsymbol{Z_1}) - U_f(a,\boldsymbol{Z_2})] \\
\nonumber
& \ \ + \cdots + [U_f(a,\boldsymbol{Z_M}) - U_f(a,\hat{\boldsymbol{Q}}(t))] = O(\| \boldsymbol{Q}(t) - \hat{\boldsymbol{Q}}(t) \|).
\end{align}
 
We can derive the same result for all other nodes in the network. In other words, there exist constants $\eta_a, \eta_b$ such that
\begin{eqnarray}
\label{InfreqClaim_06}
U_f(a,\boldsymbol{Q}(t)) &=& U_f(a,\hat{\boldsymbol{Q}}(t)) + \eta_a \| \boldsymbol{Q}(t) - \hat{\boldsymbol{Q}}(t) \|, \\
\label{InfreqClaim_07} 
U_f(b,\boldsymbol{Q}(t)) &=& U_f(b,\hat{\boldsymbol{Q}}(t)) + \eta_b \| \boldsymbol{Q}(t) - \hat{\boldsymbol{Q}}(t) \|.
\end{eqnarray}

However, (\ref{InfreqClaim_01}), (\ref{InfreqClaim_02}), (\ref{InfreqClaim_06}), and (\ref{InfreqClaim_07}) imply that $\eta_a \le \eta_b$ and
\begin{align}
\label{InfreqClaim_08}
U_f(b,\boldsymbol{Q}(t)) - U_f(a,\boldsymbol{Q}(t)) \le (\eta_b - \eta_a) \| \boldsymbol{Q}(t) - \hat{\boldsymbol{Q}}(t) \|. 
\end{align}
%Hence, (\ref{InfreqClaim_03}) holds and the proof is complete.

%Finally, using similar arguments for boundedness of outdated backlog information  used in C-ORCD optimality proof,  we
%can prove the optimality of D-ORCD. 
With this, the proof is now complete.

%\balancecolumns

\bibliographystyle{IEEEbib}
\bibliography{ORCD_journal}

\end{document}